\documentclass[a4paper,11pt]{article}
\usepackage{fullpage}

\pdfoutput=1
\usepackage[colorlinks=true,urlcolor=Blue,citecolor=Green,linkcolor=BrickRed]{hyperref}
\usepackage[usenames,dvipsnames]{xcolor}
\usepackage[utf8]{inputenc}
\usepackage{amssymb,amsmath,amsfonts,amsthm}
\usepackage{algorithm}
\usepackage{algorithmicx}
\usepackage[noend]{algpseudocode}
\usepackage[noadjust]{cite}
\usepackage[OT4]{fontenc}
\usepackage{rotating}
\usepackage{xspace}
\usepackage{verbatim}
\usepackage{mathrsfs}
\usepackage{framed}
\usepackage{thm-restate}
\usepackage[capitalise]{cleveref}
\usepackage{colortbl}
\usepackage{array,multirow}
\usepackage{authblk}
\usepackage{tikz}
\usepackage[linguistics]{forest}
\usepackage{graphicx}
\usepackage{subcaption}
\usepackage{mathtools}
\usepackage{arydshln}
\usepackage{ circuitikz }

\DeclareMathOperator{\slp}{\G}
\DeclareMathOperator{\pro}{\rho}

\newcommand{\start}{\mathcal{S}}
\newcommand{\nonterm}{\mathcal{V}}

\def\dd{\mathinner{.\,.}}
\newcommand{\Oh}{\mathcal{O}}

\newcommand{\set}[1]{\left\lbrace #1 \right\rbrace}

\DeclareMathOperator{\polylog}{polylog}

\renewcommand{\exp}{\mathrm{Exp}}
\newcommand{\h}{\mathrm{h}}
\newcommand{\w}{\mathrm{w}}

\DeclareMathOperator{\depth}{depth}
\newcommand{\G}{\mathcal{G}}
\renewcommand{\H}{\mathcal{H}}
\DeclareMathOperator{\T}{\mathcal{T}}

\DeclareMathOperator{\Scal}{\mathcal{S}}

\newcommand{\ovrt}{\ominus}
\newcommand{\ohrz}{\mathbin{\rotatebox[origin=c]{90}{$\ovrt$}}}

\DeclareMathOperator{\type}{\text{type}}
\DeclareMathOperator{\sort}{\text{sort}}
\DeclareMathOperator{\rank}{\text{rank}}
\DeclareMathOperator{\node}{\text{node}}

\newcommand{\rk}{r}
\newcommand{\N}{\mathbb{N}}
\newcommand{\Bin}{\mathsf{Bin}}
\newcommand{\ShiftBin}{\mathsf{ShiftBin}}

\newcommand{\para}[1]{\paragraph{#1}}

\newtheorem{theorem}{Theorem}[section]
\newtheorem{@theorem}[theorem]{Theorem}

\newtheorem{lemma}[theorem]{Lemma}
\newtheorem{example}[theorem]{Example}

\newtheorem{proposition}[theorem]{Proposition}
\newtheorem{definition}[theorem]{Definition}
\newtheorem{claim}[theorem]{Claim}

\crefname{conjecture}{Conjecture}{Conjectures}
\crefname{lemma}{Lemma}{Lemmas}
\crefname{problem}{Problem}{Problems}
\crefname{remark}{Remark}{Remarks}
\crefname{definition}{Definition}{Definitions}
\crefname{observation}{Observation}{Observations}
\crefname{@theorem}{Theorem}{Theorems}
\crefname{fact}{Fact}{Facts}
\crefname{claim}{Claim}{Claims} 

\title{Balancing Two-Dimensional Straight-Line Programs\footnote{Partially supported by the Polish National Science Centre grant number 2023/51/B/ST6/01505.}}

\author[1]{Itai Boneh}
\author[1]{Estéban Gabory}
\author[1]{Paweł Gawrychowski}
\author[1]{Adam Górkiewicz}
\affil[1]{Institute of Computer Science, University of Wrocław, Poland}

\date{}

\begin{document}

\maketitle

\begin{abstract}
We consider building, given a straight-line program (SLP) consisting of $g$ productions deriving a two-dimensional string
$T$ of size $N\times N$, a structure capable of providing random access to any character of $T$. For one-dimensional strings,
it is now known how to build a structure of size $\Oh(g)$ that provides random access in $\Oh(\log N)$ time. In fact,
it is known that this can be obtained by building an equivalent SLP of size $\Oh(g)$ and depth $\Oh(\log N)$
[Ganardi, Jeż, Lohrey, JACM 2021]. We consider the analogous question for two-dimensional strings: can we build an
equivalent SLP of roughly the same size and small depth?

We show that the answer is negative: there exists an infinite family of two-dimensional strings
of size $N\times N$ described by a 2D SLP of size $g$ such that any 2D SLP describing the same string of depth
$\Oh(\log N)$ must be of size $\Omega(g\cdot N/\log^{3}N)$. We complement this with an upper bound
showing how to construct such a 2D SLP of size $\Oh(g\cdot N)$. Next, we observe that one can naturally define
a generalization of 2D SLP, which we call 2D SLP with holes. We show that a known general balancing theorem by [Ganardi, Jeż, Lohrey, JACM 2021]
immediately implies that, given a 2D SLP of size $g$ deriving a string of size $N\times N$, we can construct
a 2D SLP with holes of depth $\Oh(\log N)$ and size $\Oh(g)$. This allows us to conclude that there is
a structure of size $\Oh(g)$ providing random access in $\Oh(\log N)$ time for such a 2D SLP. Further, this can be
extended (analogously as for a 1D SLP) to obtain a structure of size $\Oh(g \log^{\epsilon}N)$ providing random
access in $\Oh(\log N/\log \log N)$ time, for any $\epsilon >0$.
The same (optimal) random access time was very recently achieved by [De and Kempa, to appear in SODA 2026],
but with a significantly larger structure of size $\Oh(g \log^{2+\epsilon}N)$.
\end{abstract}

\section{Introduction}
The goal of the broad area of processing compressed strings is to design algorithms and data structures that operate
directly on the compressed representation of a string, with the complexity depending only (or mostly) on its size.
This can be of course considered for any problem and any compression method. Perhaps the most basic problem
is providing random access, that is, accessing the $i$-th character of the compressed string efficiently.
This is a fundamental prerequisite for more complex questions such as indexing.

Among the plethora of different compression methods, straight-line programs (SLPs) are particularly elegant yet powerful, and algorithms for processing SLP-compressed strings constitute an area of research on their own~\cite{DBLP:journals/gcc/Lohrey12}. Informally speaking, a SLP is simply a
context-free grammar on $g$ productions deriving a single string $T$ of length $N$.
A natural approach to providing random access to a SLP is to ensure that its depth is small, and then just
implement random access by traversing the grammar in time bounded by its depth.
While it was known how to guarantee that the depth is $\Oh(\log N)$ at the expense of increasing the size of the grammar
by a logarithmic factor~\cite{DBLP:journals/tcs/Rytter03,DBLP:journals/tit/CharikarLLPPSS05}, it was not clear if such
an increase in the size is necessary.
However, Bille et al.~\cite{DBLP:journals/siamcomp/BilleLRSSW15} settled the complexity of providing random access to a SLP
by describing a structure of size $\Oh(g)$ with query time $\Oh(\log N)$ (in the Word RAM model; see~\cite{DBLP:journals/algorithmica/BilleGGLW21}
for an alternative solution in the weaker pointer machine model).
This should be compared with a lower bound of Verbin and Yu~\cite{DBLP:conf/cpm/VerbinY13}: there exists an
infinite family of SLPs on $g$ productions describing strings of length $N=g^{1+\epsilon}$ such that, for any structure
of size $\Oh(g\polylog N)$, the query time needs to be $\Omega(\log N/\log\log N)$. Thus, the only remaining
question is whether we can design a structure of size $\Oh(g)$ and query time $\Oh(\log N/\log\log N)$.

The structure of Bille et al. requires carefully combining quite a few tools. Somewhat surprisingly, Ganardi, Jeż, and
Lohrey~\cite{DBLP:journals/jacm/GanardiJL21} showed that this is not necessary: one can always build an equivalent
SLP (describing the same string) of size $\Oh(g)$ and depth $\Oh(\log N)$. With such a grammar in hand,
random access can be directly implemented in $\Oh(\log N)$ by simply storing the length of the expansion of
each nonterminal and traversing the grammar down from the starting nonterminal. Further, one can easily obtain
a structure of size $\Oh(g\log^{\epsilon}N)$ and optimal query time $\Oh(\log N/\log\log N)$ by ``unwinding'' the first
$\Theta(\log\log N)$ levels of the derivation of each nonterminal, and storing the information about the obtained
sequence of nonterminals in a fusion tree~\cite{DBLP:conf/stoc/FredmanW90}.

In this paper, we consider a two-dimensional string, which is simply a rectangular array of characters 
of size $N\times M$. Such objects
have been studied in the area of formal languages~\cite{DBLP:reference/hfl/GiammarresiR97}.
Lempel and Ziv~\cite{DBLP:journals/tit/LempelZ86} defined a compression scheme for 2D strings by proceeding similarly as for 1D strings, and
informally speaking keeping only the first occurrence of each substring while replacing its further occurrences
by pointers to the first occurrence. However, this is not necessarily convenient for providing efficient random access,
which led Brisaboa et al.~\cite{DBLP:journals/cj/BrisaboaGGN24} to introduce 2D block trees that do allow for efficient
random access at the expense of restricting which occurrences are being replaced (and no clear bound on the size
of the compressed representation; see~\cite{DBLP:journals/access/CarfagnaM24} for an attempt at providing one).
Another natural possibility is to generalize context-free grammars to two-dimensional strings, by allowing productions
of the form $A \ohrz B$ and $A\ovrt B$, denoting the horizontal and the vertical concatenation,
respectively. For every such production, the corresponding dimensions of all strings derived from both nonterminals must be equal.
Such a generalization appears to have been introduced multiple times in the
literature~\cite{DBLP:conf/stacs/Matz97,DBLP:journals/iandc/SiromoneySK73,DBLP:conf/iwpt/Tomita89a}.
This leads to the notion of a 2D SLP, which is such a context-free grammar that derives exactly one 2D string.

Berman et al.~\cite{DBLP:journals/jcss/BermanKLPR02} showed that many problems that can be efficiently
solved on 1D SLPs become intractable on 2D SLPs. For example, compressed pattern matching
(where the pattern is uncompressed but the text is compressed) for 2D strings is NP-complete, but linear-time
for 1D SLP-compressed texts~\cite{DBLP:conf/soda/GanardiG22} or 2D RLE-compressed texts~\cite{DBLP:journals/jal/AmirBF97}.
However, for random access this is not the case: very recently, De and Kempa~\cite{de2025optimalrandomaccessconditional}
showed how to construct a structure of size $\Oh(g\log^{2+\epsilon}N)$ with query time $\Oh(\log N/\log\log N)$ (for $N\geq M$).
Carfagna et al.~\cite{carfagna2025generalizationrepetitivenessmeasurestwodimensional} observed
that, in fact, the structure of Bille et al. can be (informally speaking) applied on each dimension separately,
resulting in a structure of size $\Oh(g)$ and query time $\Oh(\log N)$.

\para{Balancing 2D SLPs.}
We study the problem of balancing 2D SLPs. Our first result is a simple upper bound, showing that given a 2D SLP
of size $g$ deriving a string of size $N\times M$ (for $N \geq M$) we can construct a 2D SLP of size $\Oh(g\cdot M)$ and depth $\Oh(\log N)$.
This is, of course, a significant increase in the size for highly compressible strings.
A natural question is whether we could, similarly to 1D SLPs, achieve the same balancing without increasing the SLP size. Namely,
given a 2D SLP on $g$ productions deriving a string of size $N\times M$ (for $N \geq M$), is there a 2D SLP on $\Oh(g)$ productions
deriving the same string but of depth $\Oh(\log N)$? We give a negative answer to this question, showing that
for infinitely many values of $N$ there exists an $N\times N$ 2D string described by a 2D SLP of size $g$, such
that any 2D SLP of depth $\Oh(\log N)$ describing the same string must consist of $\Omega(g \cdot N/\log^{3}N)$ productions.

\begin{figure}[b]
\includegraphics[width=\textwidth]{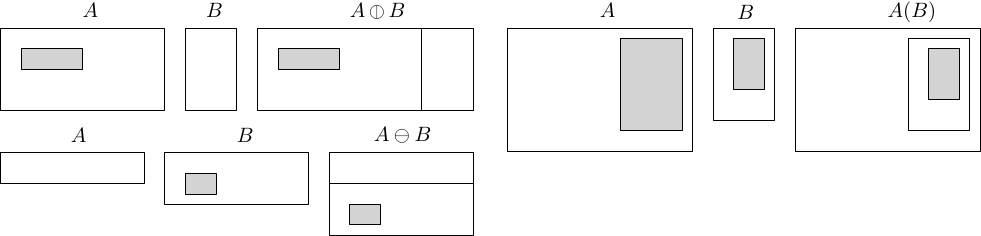}
\caption{2D SLP with holes allows horizontal concatenation $A\ohrz B$, vertical concatenation $A \ovrt B$, and substitution $A(B)$.
The gray rectangle represents the hole that might be present in at most one argument of a concatenation, and might be present in $B$
(and needs to be present in $A$) for a substitution $A(B)$. Further, the corresponding dimensions always need to match.}
\label{fig:operations}
\end{figure}

\para{2D SLPs with holes.}
To overcome this barrier, we define a natural generalization
of 2D SLPs: 2D SLPs with holes. In such a SLP, a nonterminal can also derive a 2D string from which
a rectangular fragment has been removed. One can then naturally define the productions as either horizontal
and vertical concatenations of two nonterminals in which at most one has a hole, and substituting the string (with or without a hole)
derived by a nonterminal into the hole of the string derived by another nonterminal, see~\cref{fig:operations}.
Then, given a 2D SLP with holes of depth $d$, we can implement random access in $\Oh(d)$ time
by simply storing, for every nonterminal, the size of its derived string together with the position of its hole
(if one exists). Now, given a 2D SLP of size $g$ deriving a 2D string of size $N\times M$ (for $N\ge M$), we would like to
construct a balanced 2D SLP with holes deriving exactly the same string. We observe that the general balancing theorem of
Ganardi, Jeż, and Lohrey~\cite{DBLP:journals/jacm/GanardiJL21} immediately implies that one can obtain
such a 2D SLP with holes of size $\Oh(g)$ and depth
$\Oh(\log N)$. Next, by ``unwinding'' the first $\Theta(\log\log N)$ levels of the derivation
of each nonterminal, this allows us to obtain a structure of size $\Oh(g\cdot \log^{\epsilon}N)$ and query time
$\Oh(\log N/\log\log N)$, for any $\epsilon>0$, in the Word RAM model.
Such query time is optimal for structures of size $\Oh(g \polylog N)$ by applying the lower bound for 1D SLP,
and our space is smaller by two factors of $\log N$ compared to the recent result of De and Kempa~\cite{de2025optimalrandomaccessconditional},
matching the best known bounds for 1D SLPs~\cite{DBLP:conf/esa/BelazzouguiCPT15}.

\section{Preliminaries}

\para{Integer intervals.}
For integers $i,j$, we denote $[i \dd j] = \{i,i+1, \ldots ,j \}$ ($[i \dd j]=\emptyset$ for $j<i$), and $[i]=[1 \dd i]$. 
We also denote $(i-1 \dd j+1)=[i \dd j+1)=(i-1 \dd j]=[i \dd j]$.

\para{Strings.}
Let $\Sigma$ be a fixed finite alphabet.
A string $S = S[1]S[2]\ldots S[N]$ of length $N\in \N$ is a sequence of $N$ characters from $\Sigma$.
For $i,j \in [N]$, we define a substring $X$ of $S$ as $X= S[i \dd j]=S[i]S[i+1]\ldots S[j]$.
If $i=1$ or $j=N$, the substring $X$ is called a prefix or a suffix of $S$, respectively. Given a string $S$ of length $N$ and a string $T$ of length $M$, we define their concatenation $S\cdot T=S[1]S[2]\ldots S[N]T[1]T[2]\ldots T[M]$.

\para{2D strings.}
An $N \times M$ two-dimensional (2D) string $S$ is an $N \times M$ array of characters from $\Sigma$.
We write $\h(S)$ to denote the number of rows $N$, and $\w(S)$ to denote the number of columns $M$.
We also write $\dim(S)$ to denote the pair $(N, M)$.
Similarly to 1D strings, we define a substring of $S$ as $X=S[i \dd j][i' \dd j']$ for some $i,j \in [N]$ and $i',j'\in [M]$.
For every $i\in [N]$, we define $S[i] = S[i][1 \dd M]$ as the $i$-th row of $S$.
Given two 2D strings $A$, $B$, we define their \emph{horizontal concatenation} $A \ohrz B$ when $\h(A)=\h(B)$ and their \emph{vertical concatenation} $A \ovrt B$ when $\w(A) = \w(B)$ in the natural way.
We sometimes view one-dimensional (1D) strings of length $N$ as $1 \times N$ strings, and elements of $\Sigma$ as $1 \times 1$ strings.
For a more detailed introduction on 2D strings and languages, see~\cite{DBLP:reference/hfl/GiammarresiR97}.

\para{2D SLP.}
A \emph{2D straight-line program} $\slp$ is a grammar $(\nonterm,\start,\pro)$ where: 

\begin{itemize}
    \item $\nonterm$ is a set of \emph{nonterminals}, with $\nonterm\cap \Sigma=\emptyset$. Each $X \in \nonterm$ has a \emph{dimension}
    $\dim(X)=(\h(X),\w(X))$ indicating that it derives a string of size $\h(X)\times \w(X)$.
    \item $\start\in \nonterm$ is called the \emph{starting nonterminal}.
	\item $\pro$ is a mapping on the set $\nonterm$, where $\pro(X)$ is one of the following:
    \begin{itemize}
		\item $\sigma$, where $\sigma \in \Sigma$, and $\h(X) = \w(X) = 1$,
        \item $Y \ohrz Z$, where $Y, Z \in \nonterm$, $\h(X)=\h(Y)=\h(Z)$, $\w(X)=\w(Y)+\w(Z)$,
        \item $Y \ovrt Z$, where $Y, Z \in \nonterm$, $\w(X)=\w(Y)=\w(Z)$, $\h(X)=\h(Y)+\h(Z)$.  
    \end{itemize}
    We call the pair $(X,\pro(X))$ a \emph{production}, with $X$ being its
    \emph{left-hand side} and $\pro(X)$ its \emph{right-hand side}.
    We require that the relation on $\nonterm$ defined by $X\le Y$ if $Y$ appears in the right-hand side of $X$ must be acyclic.
    We will sometimes allow a production $\pro(X)$ to be any string over $\Sigma \cup \nonterm$, and then normalize the
    grammar in the standard way.
\end{itemize}

From the acyclicity condition, it follows that for each nonterminal $X \in \nonterm$, there is a unique 2D string $\exp_{\slp}(X)$, called \emph{expansion} of $X$, derived from $X$ by recursively replacing each nonterminal in the right-hand side of its production by the corresponding derived string.
This process can be described as a tree $\T_\G$, called the derivation tree of $\G$.
Every internal node in $\T_\G$ is labeled with a nonterminal of $\G$, while each leaf is labeled with a symbol from $\Sigma$.
The root corresponds to the starting nonterminal $\start$.
The children of a node labeled with a nonterminal $X \in \nonterm$ correspond to the nonterminals that appear in $\rho(X)$, or a letter $\sigma \in \Sigma$ if $\rho(X)= \sigma$.

The 2D string derived by the SLP is denoted by $\exp(\slp) = \exp_{\slp}(\start)$.
The \emph{size} $|\slp|$ is the total number of symbols on the right-hand sides of all productions, and the \emph{depth} of $\slp$ is the depth of $\T_{\slp}$.
Notice that there is a 1-1 correspondence between the leaves of $\T_{\slp}$ and the positions in $\exp(\slp)$.

We view one-dimensional SLPs as a special case of 2D SLPs that use only horizontal concatenations and derive strings of height~$1$.
When convenient, we may interpret such a string either as a one-dimensional or $1 \times n$ two-dimensional string.

\para{Fusion trees.}
For efficient random access on 2D strings, we need need predecessor queries and fusion trees.
Let $S\subseteq [U]$ be a set of integers. A \emph{predecessor query} on $S$ takes as input an integer $x\in [U]$ and returns the largest integer $y\in S$ such that $y\le x$. We have the following result in the Word RAM model.

\begin{lemma}[\cite{DBLP:conf/stoc/FredmanW90}]\label{lem:fusion}
Given a set $S$ of $k$ integers from $[U]$, one can build in time $\Oh(k)$ a data structure of size $\Oh(k)$, called a \emph{fusion tree}, that allows predecessor queries in $\Oh(\log k / \log\log U)$ time.
\end{lemma} 

\section{Upper Bound for Balancing Standard Two-Dimensional SLPs}

The goal of this section is to prove the following theorem. For notational convenience, we assume $N\le M$, by
transposing the input we obtain the statement from the introduction.

\begin{restatable}{theorem}{restateNoHoleBalancing}\label{theorem:noholebalancing}
	Given a 2D SLP $\G$ deriving a string $T \in \Sigma^{N \times M}$ (for $N \le M$), one can construct an equivalent 2D SLP of size $\Oh(|\G| \cdot N)$ and depth $\Oh(\log M)$.
\end{restatable}

We first introduce a simple tool for when we want to concatenate multiple nonterminals.
\begin{proposition}[Concatenation gadget]\label{lem:gadget}
	Let $X_1, X_2, \dots, X_k$ be nonterminals of a 2D SLP with expansions of the same height (horizontal case) or width (vertical case).
	Then one can introduce a new nonterminal $Y$ with
	\[
		\exp(Y) =
		\begin{cases}
			\exp(X_1) \ohrz \cdots \ohrz \exp(X_k), & \text{(horizontal case),}\\[2pt]
			\exp(X_1) \ovrt \cdots \ovrt \exp(X_k), & \text{(vertical case),}
		\end{cases}
		\]
		using $\Oh(k)$ additional nonterminals arranged as a balanced binary tree.
		This increases the total size of the grammar by $\Oh(k)$ and the overall depth by $\Oh(\log k)$.
		Importantly, when multiple such gadgets are added independently (that is, their roots do not appear in each other's subtrees), the total increase in depth remains bounded by the maximum of their individual increases, not their sum.
\end{proposition}

The proof of \Cref{theorem:noholebalancing} proceeds in three steps: 
(i) turn $\G$ into a 1D SLP of size $\Oh(|\G|\cdot N)$ deriving  the concatenation of the rows; 
(ii) balance this SLP using~\cite{DBLP:journals/jacm/GanardiJL21};
(iii) reconstruct the original 2D layout while preserving logarithmic depth.

The next lemma performs the first step of the construction.

\begin{lemma}\label{lem:row-linearization}
	Given a 2D SLP $\G$ deriving a string $T \in \Sigma^{N \times M}$ (for $N \le M$), one can construct a 1D SLP of size $\Oh(|\G| \cdot N)$ deriving the concatenation of all rows of $T$, that is,
	\[
	  T[1] \ohrz T[2] \ohrz \cdots \ohrz T[N] \in \Sigma^{1 \times NM},
	\]
	where $T[i]$ denotes the $i$-th row of $T$.
\end{lemma}

\begin{proof}
	Let $\G=(\nonterm,\start,\rho)$ be the given grammar.
	For each $X \in \nonterm$ we introduce nonterminals 
	$X_1, X_2, \dots, X_{\h(X)}$.
	The goal is to construct a 1D SLP $\G'=(\nonterm',\start',\rho')$ over all the introduced nonterminals, so that for every $X \in \nonterm$ and $i \in [\h(X)]$ we have
	\[
	  \exp_{\G'}(X_i) = \exp_{\G}(X)[i].
	\]
	We first set $\nonterm' \coloneqq \set{X_i : X\in \nonterm,\, i \in [\h(X)]} \cup \set{\start'}$,
	where $\start'$ is a new starting nonterminal.
	Since each nonterminal $X$ introduces at most $\h(X) \le N$ new nonterminals, we have $|\nonterm'| = \Oh(|\nonterm| \cdot N)$.
	The productions $\rho'(X_i)$ are defined depending on the production $\rho(X)$:
	\begin{enumerate}
		\item if $\rho(X) = \sigma\in\Sigma$, then $\rho'(X_1) \coloneqq \sigma$,
		\item if $\rho(X) = A \ohrz B$, then $\rho'(X_i) \coloneqq A_i \ohrz B_i$,
		\item if $\rho(X) = A \ovrt B$ and $i \le \h(A)$, then $\rho'(X_i) \coloneqq A_i$, \label{case3}
		\item if $\rho(X) = A \ovrt B$ and $i > \h(A)$, then $\rho'(X_i) \coloneqq B_{i - \h(A)}$. \label{case4}
	\end{enumerate}
%	The unary productions in cases \ref{case3} and \ref{case4} are not in CNF, we can normalize them by removing redundant symbols.

	To define $\rho'(\start')$ as the horizontal concatenation of $\start_1, \start_2, \dots, \start_N$, we use the concatenation gadget (\cref{lem:gadget}), adding $\Oh(N)$ to the overall size of $\G'$.
	It follows by induction that for every $X \in \nonterm$ and $i \in [\h(X)]$, $\exp_{\G'}(X_i) = \exp_{\G}(X)[i]$, and thus
	$\exp_{\G'}(\start') = T[1] \ohrz \cdots \ohrz T[N]$.
\end{proof}

We now proceed with the proof of \Cref{theorem:noholebalancing}.
Let $\G$ be a 2D SLP deriving a string $T$ of size $N\times M$ (for $N \le M$).
We apply \Cref{lem:row-linearization} and then balance the resulting 1D SLP.% using the following theorem.

\begin{theorem}[{\cite[Theorem 1.2]{DBLP:journals/jacm/GanardiJL21}}]\label{thm:GJL}
	Given a 1D SLP $\G$ deriving a string of length $N$, one can construct an equivalent 1D SLP of size $\Oh(|\G|)$ and depth $\Oh(\log N)$.
\end{theorem}

We obtain a 1D SLP $\G' = (\nonterm', \start', \rho')$ of size $\Oh(|\G| \cdot N)$ and depth $\Oh(\log M)$, deriving
\[
  T[1] \ohrz T[2] \ohrz \cdots \ohrz T[N] \in \Sigma^{1 \times NM}.
\]
We view $\G'$ as a 2D SLP that uses only horizontal concatenations.
We will now modify it, so that it again derives the original 2D string $T$ while maintaining its logarithmic depth.

\begin{lemma}[Folklore]\label{lemma:folklore decomposition}
	Let $\H$ be a 1D SLP of depth $d$ deriving a string $S$.
	Any substring of $S$ can be expressed as a concatenation of $\Oh(d)$ expansions of nonterminals of $\H$.
\end{lemma}
\begin{proof}
	Take any substring from position $i$ to $j$.
	In the derivation tree of $\H$ let $u$ be the lowest common ancestor of the $i$-th and $j$-th leaf.
	Along the two root-to-leaf paths from $u$, take the siblings hanging to the right of the left path and to the left of the right path, together with the boundary leaves.
	These $\Oh(d)$ nonterminals expand to the substring.
\end{proof}

Since every row $T[i]$ is a substring of $\exp(\G')$, we can use \Cref{lemma:folklore decomposition} to identify $\Oh(\log M)$ nonterminals $X_1, X_2, \dots, X_k$, such that
\[
  T[i] = \exp_{\G'}(X_1) \ohrz \exp_{\G'}(X_2) \ohrz \cdots \ohrz \exp_{\G'}(X_k).
\]
Then for each $i$, we introduce a new nonterminal $R_i$ by applying \cref{lem:gadget} to $X_1,\dots,X_k$, implementing the production
\[
  \rho'(R_i) \coloneqq X_1 \ohrz X_2 \ohrz \cdots \ohrz X_k.
\]
Each gadget adds $\Oh(\log M)$ depth locally and $\Oh(\log M)$ symbols, but since the $R_i$'s are created independently (none of them occurs in another’s subtree), the overall depth of the grammar increases by only an additive $\Oh(\log \log M)$ term.
The total increase in size is $\Oh(N \log M)$, which is $\Oh(|\G| \cdot N)$, because 
of the straightforward lower bound of the $\Omega(\log M)$ on the size of $\G$.

Finally, we introduce a new starting nonterminal $\start''$ with the production
\[
  \rho'(\start'') \coloneqq R_1 \ovrt R_2 \ovrt \cdots \ovrt R_N,
\]
constructed using \cref{lem:gadget}, adding $\Oh(N)$ to the size and $\Oh(\log N)$ to the depth of $\G'$.
Then
\[
  \exp_{\G'}(\start'') = \exp_{\G'}(R_1) \ovrt \cdots \ovrt \exp_{\G'}(R_N) = T[1] \ovrt \cdots \ovrt T[N] = T.
\]
The resulting grammar has size $\Oh(|\G| \cdot N)$, and depth $\Oh(\log M)$.

\section{Lower Bound for Balancing Without Holes}
\label{sec:lb}
In this section, we prove that the linear multiplicative blowup of \cref{theorem:noholebalancing} is essentially unavoidable in the general case.
Formally, we prove the following.
\begin{theorem}\label{lem:lowerbound}
     For every integer $c$, there are infinitely many integers $N$ for which there exists a 2D SLP $\G$ with $g$ nonterminals deriving a string $S\in\Sigma^{N\times N}$ such that any 2D SLP deriving $S$ has either depth at least $c\log N$ or $\Omega(g \cdot \frac{N}{\log^3N})$ nonterminals.\end{theorem}
We begin by presenting some gadgets that will be useful for proving \cref{lem:lowerbound}. Intuitively, our goal is to
obtain a highly compressible 2D string such that most of its rows correspond to very non-compressible 1D strings.
Our construction can be seen as a stronger variant of the gadget used by Berman et al.~\cite[Theorem 2.2]{DBLP:journals/jcss/BermanKLPR02}
(using their gadget directly would only allow us to obtain a lower bound of roughly $\Omega(g\cdot \sqrt{N})$ in \cref{lem:lowerbound}).

\begin{definition}[See \cref{fig:Examples1}]
	For every $N=2^{n}$, we define $\Bin_N$ as the 2D string with dimensions $ N\times (n + 2) $ in which the $i$-th row is the $n$-bit binary representation of the integer $i-1$, surrounded by $\$$ symbols.
Next, we define $\ShiftBin_N$ as the binary 2D string with dimensions $2N \times N (n +2)$ such that for every $i\in [0 \dd N-1]$, the substring $\ShiftBin_N[1 \dd 2N][i\cdot (n+2)+1  \dd  (i+1)\cdot (n+2) ]$ is a copy of $\Bin_N$ `shifted' downwards by $i-1$ rows.
\end{definition}

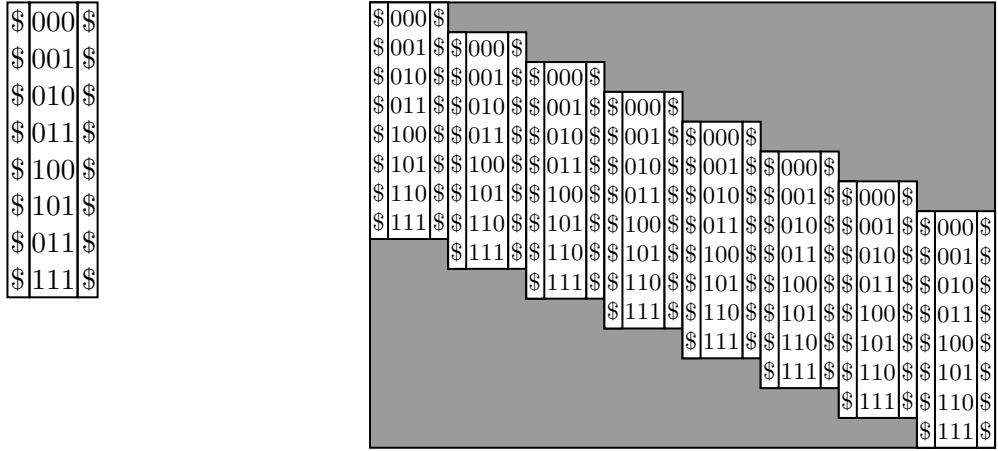
\begin{figure}[h!]
    \centering
    \begin{subfigure}[T]{0.3\textwidth}
    \begin{center}
        \tikzset{every picture/.style={line width=0.75pt}} %set default line width to 0.75pt        
\begin{tikzpicture}[x=0.75pt,y=0.75pt,yscale=-1,xscale=1]
 \begin{scope}[shift={(0, 0)}]

%Shape: Rectangle [id:dp07060328848182795] 
\draw   (288,56.67) -- (299,56.67) -- (299,205) -- (288,205) -- cycle ;
%Shape: Rectangle [id:dp526157807378003] 
\draw   (299,56.67) -- (323,56.67) -- (323,205) -- (299,205) -- cycle ;
%Shape: Rectangle [id:dp7106493519984731] 
\draw   (323,56.67) -- (333,56.67) -- (333,205) -- (323,205) -- cycle ;

% Text Node
\draw (288,58.67) node [anchor=north west][inner sep=0.75pt]   [align=left] {$ \displaystyle \$\:000\:\$ $};
% Text Node
\draw (288,77.21) node [anchor=north west][inner sep=0.75pt]   [align=left] {$\displaystyle \$\:001\:\$ $};
% Text Node
\draw (288,169.91) node [anchor=north west][inner sep=0.75pt]   [align=left] {$\displaystyle \$\:011\:\$ $};
% Text Node
\draw (288,151.37) node [anchor=north west][inner sep=0.75pt]   [align=left] {$\displaystyle \$\:101\:\$ $};
% Text Node
\draw (288,95.75) node [anchor=north west][inner sep=0.75pt]   [align=left] {$\displaystyle \$\:010\:\$ $};
% Text Node
\draw (288,132.83) node [anchor=north west][inner sep=0.75pt]   [align=left] {$\displaystyle \$\:100\:\$ $};
% Text Node
\draw (288,114.29) node [anchor=north west][inner sep=0.75pt]   [align=left] {$\displaystyle \$\:011\:\$ $};
% Text Node
\draw (288,188.45) node [anchor=north west][inner sep=0.75pt]   [align=left] {$\displaystyle \$\:111\:\$ $};
\end{scope}
\end{tikzpicture}
        \end{center}
        \label{fig:BinExample}
    \end{subfigure}
    \hfill
    \begin{subfigure}[T]{0.6\textwidth}
        \tikzset{every picture/.style={line width=0.75pt}} %set default line width to 0.75pt        

\begin{tikzpicture}[x=0.75pt,y=0.75pt,yscale=-1,xscale=1]
%uncomment if require: \path (0,312); %set diagram left start at 0, and has height of 312
\draw  [fill={rgb, 255:red, 155; green, 155; blue, 155 } ] (7,0) -- (319,0) -- (319,224) -- (7,224) -- cycle ;

  \foreach \i in {0,...,7} {
    % Draw the same rectangle, each time shifted by (xshift, yshift)
    \begin{scope}[shift={(\i*39, \i*15)}]
      \draw  [fill={rgb, 255:red, 255; green, 255; blue, 255 }  ,fill opacity=1 ] (7,0) -- (37,0) -- (37,119) -- (7,119) -- cycle ;
%Shape: Rectangle [id:dp7106493519984731] 
\draw  [fill={rgb, 255:red, 255; green, 255; blue, 255 }  ,fill opacity=1 ] (37,0) -- (46,0) -- (46,119) -- (37,119) -- cycle ;
%Shape: Rectangle [id:dp07060328848182795] 
\draw [fill={rgb, 255:red, 255; green, 255; blue, 255 }  ,fill opacity=1 ]  (7,0) -- (16,0) -- (16,119) -- (7,119) -- cycle ;

% Text Node
\draw (0,0) node [anchor=north west][inner sep=0.75pt]  [font=\footnotesize] [align=left] {$\displaystyle  \begin{array}{{>{\displaystyle}l}}
\$\:000\:\$\\
\$\:001\:\$\\
\$\:010\:\$\\
\$\:011\:\$\\
\$\:100\:\$\\
\$\:101\:\$\\
\$\:110\:\$\\
\$\:111\:\$
\end{array}$};
    \end{scope}
}

\end{tikzpicture}
        \label{fig:ShiftBinExample}
    \end{subfigure}
    \caption{Left: $\Bin_{8}$, right: $\ShiftBin_8$ (the gray area is all $0$'s).}
    \label{fig:Examples1}
\end{figure}

We observe that $\Bin_N$ and (more importantly) $\ShiftBin_N$ are highly compressible.
\begin{lemma}
For every $N=2^n$, there is a 2D SLP with $\Oh(n)$ nonterminals deriving $\ShiftBin_{N}$.
\end{lemma}
\begin{proof}
    For $N=2^{n}$, let $\Bin'_N$ denote the substring of $\Bin_N$ obtained by removing the $\$$ symbols on the margins.
    We start by showing that $\Bin'_{N}$ can be derived with a grammar of size $\Oh(n)$.
    Since $\Bin'_2$ has constant size, it can be derived using a constant number of nonterminals;
	let the starting nonterminal of the grammar be $\start_1$.
    We also construct nonterminals $0_0$ and $1_0$ that produce a single $0$ and a single $1$, respectively. 
    To construct $\Bin'_{2^{i+1}}$ for $i\ge 1$, we inductively define nonterminals $\start_{i+1}$, $\start^0_{i+1}$, $\start^1_{i+1}$, $0_{i+1}$, and $1_{i+1}$ with the productions
    \begin{enumerate}
    \item $\rho(0_{i+1})   \coloneqq 0_{i}\ovrt 0_{i}$,
    \item $\rho(1_{i+1})   \coloneqq 1_{i} \ovrt 1_{i}$,
    \item $\rho(\start^0_{i+1}) \coloneqq 0_{i+1} \ohrz \start_{i}$,
    \item $\rho(\start^1_{i+1}) \coloneqq 1_{i+1} \ohrz \start_{i}$,
    \item $\rho(\start_{i+1})   \coloneqq \start^0_{i+1} \ovrt \start^1_{i+1}$
    \end{enumerate}
	where we inductively assume that $\exp(\start_i) = \Bin'_{2^i}$.
    It can be verified that $0_{i+1}$ (resp. $1_{i+1}$) expands to a column of $0$ symbols (resp. $1$ symbols) with height $2^{i}$.
    Then, $\start^0_{i+1}$ expands to a row of $2^i$ zeros followed by $\Bin'_{2^i}$, and $\start^1_{i+1}$ expands to a row of $2^i$ ones followed by $\Bin'_{2^i}$. 
    The vertical concatenation of $\start^0_{i+1}$ and $\start^1_{i+1}$ expands to exactly $\Bin'_{2^{i+1}}$, as required.
    We introduce $\Oh(1)$ new nonterminals in order to extend the grammar of $\Bin'_{2^i}$ to the grammar of $\Bin'_{2^{i+1}}$, so the total number of nonterminals for $\Bin'_{2^n}= \Bin'_{N}$ is $\Oh(n)$. 
    The $\$$ symbols on the margins can be added using $\Oh(n)$ extra nonterminals. Hence, we can construct a grammar with $\Oh(n)$ nonterminals deriving $\Bin_{N}$.
    
    We proceed to describe a 2D SLP for $\ShiftBin_{N}$.
    Starting with $\start_n$, we inductively construct nonterminals that expand to shifted concatenations of $\Bin_{N}$.
    Formally, we let $A_0 \coloneqq \start_n$, and for every $i>0$ we define $A_i$, $A^0_i$, $A^1_i$, and $Z_i$ such that:
    \begin{enumerate}
		\item $Z_i$ expands to a matrix of $0$'s with dimensions $2^{i-1}\times 2^{i-1}\cdot n+2$,
        \item $\rho(A^0_i) \coloneqq A_{i-1} \ovrt Z_{i}$,
        \item $\rho(A^1_i) \coloneqq Z_i \ovrt A_{i-1}$,
        \item $\rho(A_i)   \coloneqq A^0_i \ohrz A^1_i$.
    \end{enumerate}
   
    It can be verified by induction that $A_{i}$ expands to a sequence of $2^i$ copies of $\Bin_N$, each copy shifted by 1 relatively to the copy to its left.
    In particular, $A_n$ expands to exactly $\ShiftBin_{N}$.
    We only introduced $\Oh(n)$ nonterminals on top of the nonterminals $\start_n$ to construct $A_n$ (note that all $Z_i$ nonterminals can be constructed using additional $\Oh(n)$ nonterminals), so the total size of the grammar deriving $\ShiftBin_{N}$ is $\Oh(n)$ as required.
\end{proof}

Even though $\ShiftBin_N$ is highly compressible, it contains very incompressible 1D strings.
In the following lemma, we show that the central rows of $\ShiftBin_N$ (i.e., rows that are a constant fraction of the height of $\ShiftBin_N$ away from the top and bottom) are highly incompressible.
\begin{lemma}\label{lem:shiftbincut}
    For every $N=2^n$, and $i\in [2N]$, denote as $R_{N,i}$ the $i$-th row of $\ShiftBin_{N}$.
	Any (1D) SLP deriving a superstring of $R_{N,i}$ is of size at least $\min\{i, 2N-i+1\}$. 
\end{lemma}
\begin{proof}
Observe that the string $R_{N,i}$ contains substrings of the form $\$ b \$$ where $b$ is the $\log N$-bit binary encoding of a number $z_b$.
Specifically, $R_{N,i}$ contains such substrings for $\min\{i,2N-i+1\}$ distinct values of $z_b$.

Let $\G$ be a grammar deriving a superstring of $R_{N,i}$. 
For every $z_b$ such that $\$ b \$ $ is a substring of $R_{N,i}$, denote as $A(z_b)$ the minimal nonterminal of $\G$ that expands to a superstring of $\$ b\$$.
That is, $A(z_b)$ is a nonterminal expanding to a string containing $\$ b\$$ as a substring, with rule $\rho(A(z_b))= B \cdot C$ such that neither $B$ nor $C$ contains $\$ b\$$ as a substring.

We claim that for every two distinct $z_b,z_{b'}$, it holds that $A(z_b) \neq A(z_{b'})$.
Indeed, $\rho(A(z_b))= B \cdot C$ such that $\exp(B)$ has a suffix of the form $\$b_1$ and $\exp(C)$ has a prefix of the form $b_2\$$ satisfying $b = b_1\cdot b_2$.
Similarly, $\rho(A(z_{b'}))= B'\cdot C'$ such that $B'$ has suffix $\$b'_1$ and $C'$ has prefix $b'_2\$$ satisfying $b' = b'_1\cdot b'_2$.
Since $b\neq b'$ and $\$$ does not occur in $b$ or $b'$, it cannot be the case that $A(z_b) = A(z_{b'})$. 
We have shown that $\G$ contains a distinct nonterminal for every one of the $\min\{i, 2N-i+1\}$ different values of $z_b$, which concludes the proof.
\end{proof}

Next, we show that a 2D grammar deriving a string $S$ implies a 1D grammar of similar size deriving the margins of $S$.
\begin{lemma}\label{obs:2dframe1d}
	Let $S$ be a 2D string and let $t$, $b$ (resp. $\ell$, $r$) be the top and bottom rows (resp. left and right columns) of $S$.
	Let $\G = (\nonterm, \start, \rho)$ be a 2D SLP deriving $S$.
There are 1D SLPs deriving $t$, $b$, $\ell$, and $r$, each of size at most $|\G|$.
\end{lemma}
\begin{proof}
    We show how to construct a grammar for $t$, the remaining cases are symmetric.
    
    We create 1D SLP $\G'$ with nonterminals $\nonterm' = \{A' : A \in \nonterm \}$.
	For every nonterminal $A \in \nonterm$, if $\rho(A)= B \ovrt C$, we set $\rho(A') \coloneqq B'$ (i.e., keeping just the top part of $\exp(A)$).
    If $\rho(A) = B \ohrz C$, we set $\rho(A') \coloneqq B' \ohrz C'$.
    If $\rho(A)=\sigma$ for some terminal $\sigma$, we set $\rho(A') \coloneqq \sigma$.
    
    It is clear that $|\G'| \le |\G|$, and it can be verified (for example, by induction) that $\G'$ derives the top row of $S$, as required.
\end{proof}

Let us define the main gadget used for proving \cref{lem:lowerbound}.

\begin{definition}[See \cref{fig:gadgetexample}]
    For every $N,M$, we define a 2D string $C_{N,M}$ with dimensions $N \times M$.
    Choose the largest $M'=2^{n}$ such that $M' (\log M' + 2) \le M/2$ and let $B=\ShiftBin_{M'}$.
    $C_{N,M}$ contains two vertical blocks: 
    The left block is formed by $\lfloor \frac{N}{2M'} \rfloor$ vertically concatenated copies of $B$.
    The right block is formed by $\lfloor \frac{N - M'}{2M'} \rfloor$ vertically concatenated copies of $B$, shifted vertically by $M'$ positions.  
    Each block is padded with $0$ symbols so that its total height equals $N$.  
    Finally, to the right of the blocks there is a padding of $0$ symbols to complete the width to $M$.
\end{definition}

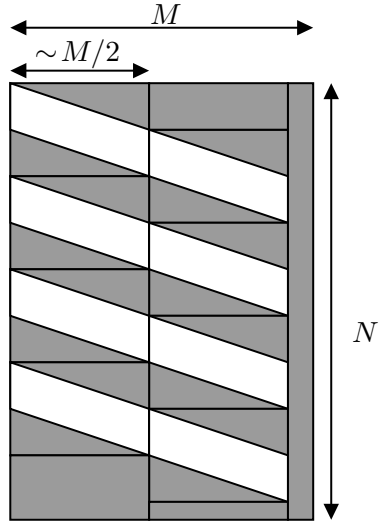
\begin{figure}[ht] 
    \centering
     \tikzset{every picture/.style={line width=0.75pt}} %set default line width to 0.75pt        

\begin{tikzpicture}[x=0.75pt,y=0.75pt,yscale=-1,xscale=1,scale=0.9]
%uncomment if require: \path (0,398); %set diagram left start at 0, and has height of 398

\foreach \i in {0,...,3} {
    \foreach \j in {0,1}{
 
    \begin{scope}[shift={( 77* \j, \i*52 + 26* \j)}]

%Shape: Rectangle [id:dp1665119391968557] 
\draw [fill={rgb, 255:red, 155; green, 155; blue, 155 }  ,fill opacity=1 ]    (157,81) -- (234,81) -- (234,133) -- (157,133) -- cycle ;
%Shape: pararelloid [id:dp1665119391968557] 
\draw [fill={rgb, 255:red, 255; green, 255; blue, 255 }  ,fill opacity=1 ]    (157,81) -- (234,107) -- (234,133) -- (157,107) -- cycle ;

\end{scope}
}
}

% Top Padding
\draw [fill={rgb, 255:red, 155; green, 155; blue, 155 }  ,fill opacity=1 ] (234,81)--(311,81)--(311,107)--(234,107)--cycle;
% Bot Padding 1
\draw [fill={rgb, 255:red, 155; green, 155; blue, 155 }  ,fill opacity=1 ] (157,289)--(234,289)--(234,325)--(157,325)--cycle;
% Bot Padding 2
\draw [fill={rgb, 255:red, 155; green, 155; blue, 155 }  ,fill opacity=1 ] (234,315)--(311,315)--(311,325)--(234,325)--cycle;
% Right Padding
\draw [fill={rgb, 255:red, 155; green, 155; blue, 155 }  ,fill opacity=1 ] (311,81)--(325,81)--(325,325)--(311,325)--cycle;

%Horizontal Arrow 1 [id:da6304654702063035] 
\draw    (160,50) -- (323,50) ;
\draw [shift={(325,50)}, rotate = 180] [fill={rgb, 255:red, 0; green, 0; blue, 0 }  ][line width=0.08]  [draw opacity=0] (8.93,-4.29) -- (0,0) -- (8.93,4.29) -- cycle    ;
\draw [shift={(157,50)}, rotate = 0] [fill={rgb, 255:red, 0; green, 0; blue, 0 }  ][line width=0.08]  [draw opacity=0] (8.93,-4.29) -- (0,0) -- (8.93,4.29) -- cycle    ;

%Horizontal Arrow 2 [id:da6304654702063035] 
\draw    (159,74) -- (232,74) ;
\draw [shift={(234,74)}, rotate = 180] [fill={rgb, 255:red, 0; green, 0; blue, 0 }  ][line width=0.08]  [draw opacity=0] (8.93,-4.29) -- (0,0) -- (8.93,4.29) -- cycle    ;
\draw [shift={(157,74)}, rotate = 0] [fill={rgb, 255:red, 0; green, 0; blue, 0 }  ][line width=0.08]  [draw opacity=0] (8.93,-4.29) -- (0,0) -- (8.93,4.29) -- cycle    ;

% Vertical Arrows [id:da48174428118910584] 
\draw    (335,83) -- (335,323) ;
\draw [shift={(335,325)}, rotate = 270] [fill={rgb, 255:red, 0; green, 0; blue, 0 }  ][line width=0.08]  [draw opacity=0] (8.93,-4.29) -- (0,0) -- (8.93,4.29) -- cycle    ;
\draw [shift={(335,81)}, rotate = 90] [fill={rgb, 255:red, 0; green, 0; blue, 0 }  ][line width=0.08]  [draw opacity=0] (8.93,-4.29) -- (0,0) -- (8.93,4.29) -- cycle    ;

% Text Node
\draw (169,55) node [anchor=north west][inner sep=0.75pt]   [align=left] {$\displaystyle \sim\!M/2$};
% Text Node
\draw (233,35) node [anchor=north west][inner sep=0.75pt]   [align=left] {$\displaystyle M$};
% Text Node
\draw (345,213) node [anchor=north west][inner sep=0.75pt]   [align=left] {$\displaystyle N$};

\end{tikzpicture}
        \caption{A demonstration of $C_{N,M}$.
         Each rectangle with a diagonal pattern is a copy of $\ShiftBin_{M'}$ (the diagonal pattern corresponds to the shifted copies of $\Bin_{M'}$ within $\ShiftBin_{M'}$).
         The gray areas are the padding $0$ used to ``shift'' the second vertical block, and the extra vertical padding at the bottom to complement the height of each vertical block to $N$.
        }
        \label{fig:gadgetexample}
\end{figure}

Recall that the width of $\ShiftBin_{M'}$ is $M' (\log M' + 2)$.
Therefore, the width of each vertical block is $M' (\log M' + 2) \le M/2$, so their total width never exceeds $M$.
Also recall that the height of $\ShiftBin_{M'}$ is $2M'$, so each block contains the maximal number of vertically concatenated copies of $\ShiftBin_{M'}$ that fit within height $N$ (for the right block, this takes into account the initial vertical shift of $M'$ positions).

We proceed to show that there is a logarithmic sized 2D SLP deriving $C_{N,M}$.
\begin{lemma}\label{obs:gadgetsconstruction}
    For every $N,M$, there is a 2D SLP of size $\Oh(\log N + \log M)$ deriving $C_{N,M}$.
\end{lemma}
\begin{proof}
  Each of the two blocks of $C_{N,M}$ can be produced by starting with a grammar of size $\Oh(\log M)$ that derives $B$, and adding $\Oh(\log (N/M))$ nonterminals that produce concatenations of exponentially increasing amounts of $B$.
    Then, to get a block of $C_{N,M}$ we can concatenate $\Oh(\log (N/M))$ such nonterminals.
    The $0$ paddings can be obtained using additional $\Oh(\log N + \log M)$ nonterminals.
\end{proof}

In our lower bound construction, we are interested in constructing a grammar with nonterminals producing a sequence of $k+1$ gadgets of the form 
$$C_{N,M},C_{N+b,M},C_{N+2b,M},\ldots, C_{N+k\cdot b,M}$$
for some $b$ that is an integer multiple of $M'$.
In the following lemma, we show that nonterminals producing such sequence can be constructed using a small amount of additional nonterminals.
\begin{lemma}\label{lem:genseqgad}
Let $N, M \in \N$ and choose the largest $M'=2^{n}$ such that $M'(\log M' + 2) \le M/2$.
Let $b\le N$ be some integer multiple of $M'$ and let $k\in \N$.
There is a grammar $\G$ satisfying the following.
\begin{enumerate}
    \item For every $i \in [0 \dd k]$, $\G$ contains a nonterminal $C_i$ that expands to $C_{N + i \cdot b, M}$.
    \item $\G$ contains $\Oh(\log N + \log M + k)$ nonterminals in total.
\end{enumerate} 
\end{lemma}
\begin{proof}
    We will show how to construct a grammar with the required amount of nonterminals, containing $C_i$ for even values of $i$.
    A similar grammar can be constructed for odd values of $i$, and the union of the two grammars satisfies the required conditions.
    For an integer $N'$, we denote the following substrings of $C_{N',M}$:
    \begin{enumerate}
    \item $B^1_{N'}$ is the left vertical block of $C_{N',M}$, \textit{without} the vertical padding to complement the height to $N$.
    \item $B^2_{N'}$ is defined identically to $B^1_{N'}$, but for the right block of $C_{N',M}$.
    \item $B^3_{N'}$ is the horizontal padding to complement the width of $C_{N',M}$ to $M$.
    \item $P^1_{N'}$ is the vertical padding at the bottom of the left vertical block.
    \item $P^2_{N'}$ is the vertical padding at the bottom of the right vertical block.
    \end{enumerate}

    We observe that $P^1_N$ and $P^2_N$ are all-$0$'s strings, and their dimensions depend only on the remainder of $N$ modulo $2M'$. 
    Therefore, $P^{1}_N = P^1_{N+2b}$ and $P^{2}_N = P^2_{N + 2b}$.
    We also observe that $B^1_{N + 2b}$ is obtained from $B^1_{N}$ by appending $b/M'$ copies of $\ShiftBin_{M'}$ to the bottom of $B^1_{N}$, and the same is true for $B^2_{N}$ and $B^2_{N + 2b}$.
    Finally, $B^3_{N + 2b}$ is obtained from $B^3_N$ by appending a block of $0$'s with height $2b$ and appropriate width.

    Let $\G_{N,M}$ be the grammar deriving $C_{N,M}$ as described in \cref{obs:gadgetsconstruction}.
    Recall that the grammar $\G_{N,M}$ has nonterminals generating each of $B^1_N$, $B^2_N$, $B^3_N$, $P^1_N$ and $P^2_N$.
    We extend $\G_{N,M}$ to include the following nonterminals:
    \begin{enumerate}
        \item $V'_1$: a vertical concatenation of $b/M'$ copies of $\ShiftBin_{M'}$.
        \item $V'_2$: an all-$0$'s block with height $2b$ and width $w(B^3_N)$.
    \end{enumerate}
    Observe that both $V'_1$ and $V'_2$ can be added to the grammar using $\Oh(\log N + \log M)$ additional nonterminals.
    It follows from the above discussion that if we have a grammar with nonterminals expanding to $B^1_{N'}$, $B^2_{N'}$, $B^3_{N'}$, $P^1_{N'}$ and $P^2_{N'}$, we can add $\Oh(1)$ nonterminals to obtain nonterminals expanding to $B^1_{N' + 2b}$, $B^2_{N' + 2b}$, $B^3_{N' + 2b }$, $P^1_{N' + 2b}$, and $P^2_{N' + 2b}$.
    Then, a nonterminal $C_{N' + 2b}$ can be added by adding additional $\Oh(1)$ nonterminals concatenating the above appropriately.
    By repeating this process $\lfloor k/2 \rfloor$ times, we obtain a grammar containing $C_{N + i \cdot b, M}$ for every even $i$ in $[0\dd k]$.
    The total number of nonterminals is $\Oh(\log N + \log M + k)$, as required.
\end{proof}

We are now ready to prove \Cref{lem:lowerbound}.

\begin{proof}[Proof of \cref{lem:lowerbound}]
    Let us start by providing some intuition.
    We will construct a string $S$ that is composed of instances of $C_{N',M'}$ with $M'\in \Theta(\frac{N}{\log N})$.
    The instances are nested in a spiral manner with decreasing lengths (see \cref{fig:lb}).
    The 'depth' of this spiral (i.e., number of swirls  before the center is reached) will be roughly $c\log N$.
    Since every 'step' in this pattern is a $C_{N',M'}$ gadget with $N',M' \in \Oh(N)$, a grammar $\G$ deriving $S$ can be constructed using $\Oh(\log^2N)$ nonterminals.

    To understand the intuition behind this construction, think of a 2D SLP $\G'$ deriving $S$ with a small number of nonterminals (i. e. , $|\G'| \in o(\frac{N}{\log^2 N})$ ).
    Think of a path $P$ in the derivation tree of $\G'$ from the root to some index at the center of $S$.
    It is helpful to think of $P$ as a sequence of nested rectangles in $S$, all containing the center point of $S$.
    Intuitively, each of the $C_{N',M'}$ gadgets limits the structure of these rectangles.
    Namely, these rectangles cannot horizontally intersect a $C_{N',M'}$ gadget (except possibly very close to the top/bottom of said gadget), as this would imply that the top/bottom of a rectangle contains some row of a $\ShiftBin$ gadget as a substring.
    By \cref{lem:shiftbincut} and \cref{obs:2dframe1d}, a nonterminal that expands to a string with such a top/bottom row requires a large number of nonterminals, contradicting the small size of $\G'$.
    It follows that the rectangles of $P$ have to roughly comply with the structure of the spiral.
    This allows us to lower bound $|P|$, and thus the depth of $\G'$.

    We proceed to define $\G$ and $S$.
    Let $\Delta'=\frac{N}{8c\log N}$ and $\lambda = 2c \log N $.
    Let $M'$ be the largest power of $2$ such that $M'(\log M'+2) \le \Delta'/2$.
    Let $\Delta$ be the largest integer smaller than $\Delta'$ that is an integer multiple of $M'$.
    Notice that we must have $\Delta \in [\Delta'/2 \dd \Delta']$ since $\Delta' \ge 2M'$ and $\Delta \ge \Delta' - M'$. 
    Also observe that $M''$, the largest power of two such that $M''(\log M'' +2) \le \Delta$, must be exactly $M''=M'$ (otherwise, $2M'(\log M' +2)$ is an integer multiple of $M'$ strictly between $\Delta'$ and $\Delta$).

    For an integer $N'\in \N$, let $\tilde{C}_{N',\Delta}$ denote the string $C_{N',\Delta}$ rotated clockwise by $90$ degrees. 
    We denote as $G(N')$ and $G(N')$ nonterminals that expand to $C_{N',\Delta}$ and $\tilde{C}_{N',\Delta}$, respectively.

    We define the SLP $\G=(\nonterm,F^0_0,\rho)$ as follows.
    For every $i\in [0 \dd \lambda -1]$ and $x\in \{0,1,2,3\}$ there is a nonterminal $F^x_i$ in $\nonterm$ with the following derivation rules (see \cref{fig:lb} for a demonstration):
    \begin{enumerate}
        \item $\rho(F^0_i)   \coloneqq F^1_i \ohrz G(N - 2 i\cdot \Delta)$ for every $i\in [0 \dd \lambda-1]$,
        \item $\rho(F^1_i)   \coloneqq F^2_i \ovrt \tilde{G}(N - (2i+1)\cdot \Delta)$ for every $i\in [0 \dd \lambda-1]$,
        \item $\rho(F^2_{i}) \coloneqq G(N - (2i+1)\cdot \Delta) \ohrz F^3_{i}$ for every $i\in [0 \dd \lambda-1]$,
        \item $\rho(F^3_i)   \coloneqq \tilde{G}(N - (2i+2)\cdot \Delta) \ovrt F^0_{i+1}$ for every $i\in [0 \dd \lambda-2]$.
    \end{enumerate}

    \begin{figure}[ht]
  \centering
  \tikzset{every picture/.style={line width=0.75pt}} %set default line width to 0.75pt        

\begin{tikzpicture}[x=0.75pt,y=0.75pt,yscale=-1,xscale=1,scale=0.8]
	\usetikzlibrary{fit}
	\draw (161,26) -- (461,26) -- (461,326) -- (161,326) -- cycle;
	\draw (411,26) -- (461,26) -- (461,326) -- (411,326) -- cycle;
	\draw (161,276) -- (411,276) -- (411,326) -- (161,326) -- cycle;
	\draw (161,26) -- (211,26) -- (211,276) -- (161,276) -- cycle;
	\draw (211,26) -- (411,26) -- (411,76) -- (211,76) -- cycle;
	\draw (361,76) -- (411,76) -- (411,276) -- (361,276) -- cycle;
	\draw (211,76) -- (261,76) -- (261,226) -- (211,226) -- cycle;
	\draw (211,226) -- (361,226) -- (361,276) -- (211,276) -- cycle;
	\draw (261,76) -- (361,76) -- (361,126) -- (261,126) -- cycle;
	\draw  [fill={rgb, 255:red, 211; green, 211; blue, 211 }, fill opacity=1] (261,126) -- (361,126) -- (361,226) -- (261,226) -- cycle ;

	\node[rotate=90] at (436,176) {$C_{N,\Delta}$}; % center of (411,26)-(461,326)
	\node[]          at (286,301) {$\tilde{C}_{ N-\Delta,\Delta}$};  % center of (161,276)-(411,326)
	\node[rotate=90] at (186,151) {$C_{N-\Delta,\Delta}$};  % center of (161,26)-(211,276)
	\node[]          at (311, 51) {$\tilde{C}_{N-2\Delta,\Delta}$};   % center of (211,26)-(411,76)
\end{tikzpicture}
  \caption{A demonstration of the string $S$ obtained by the described grammar.
		The grammar composes rotated copies of $C_{N,M}$ in a spiral pattern.
        The narrow dimension of each gadget is always $\Delta = \frac{N}{8c\log N}$, and the wide dimension of the gadgets decreases as the spiral swirls inwards.
        }
  \label{fig:lb}
\end{figure}
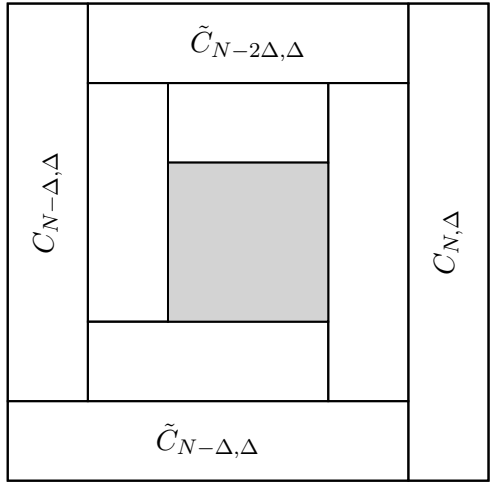

	Finally, $F^3_{\lambda-1}$ expands to a string of size $N-4\lambda\Delta \times N-4\lambda \Delta$ filled with $0$'s (notice that $N - 4\lambda\Delta \ge N-4\lambda\Delta' = N/2$).
	The starting nonterminal of $\G$ is $F^0_0$, which expands to the string $S$.
    In words, the nonterminals alternate between expanding to a $C_{N',\Delta}$ gadget, and expanding to a rotated $C_{N',\Delta}$ gadget.
    The width (or height, if rotated) of the gadgets is always $\Delta$, and the height $N'$ (if rotated, width) decreases to fit the remaining height/width of $S$.

	It follows from \cref{lem:genseqgad} that all $G$ and $\tilde{G}$ nonterminals used in the above can be created using $\Oh(\log N + \log \Delta + \lambda) = \Oh(\log N)$ nonterminals.
    Clearly, the all-$0$'s string of $F^3_{\lambda-1}$ can also be produced using $\Oh(\log N)$ nonterminals.
    The total number of nonterminals in the grammar is therefore $n = \Oh(\log N)$.
    For every $i\in [0 \dd \lambda-1]$, we call the string produced by $F^0_i$ the $i$-th layer of $S$, denoted by $L_i$.
    In terms of indices in $S$, the $i$-th layer of $S$ is defined as $L_i = S(i\Delta \dd N-i\Delta](i\Delta \dd N-i\Delta]$.
    The $\Delta$ right columns of $L_i$ are an occurrence of $C_{N',\Delta}$ with $N'=N-2i\Delta$ being the height of $L_i$.

    We proceed to prove that no grammar $\G'$ with $o(N/\log^2 N)$ nonterminals and depth less than $c \log N$ can derive $S$.
    Assume to the contrary that such a grammar $\G'=(\nonterm',\start',\rho')$ exists.
	For any nonterminal $X\in \nonterm'$ and $i\in [\lambda]$, we say that the nonterminal $X$ contains $L_i$ if it expands to $S[x \dd x'][y \dd y']$ such that $(i\Delta \dd N-i\Delta] \subseteq [x \dd x']\cap [y \dd y']$, i.e., $\exp_{\G'}(X)$ is a superstring of $L_{i}$.
    Let $P = V_1,V_2,\ldots V_{d}$ be the sequence of nonterminals obtained by traversing the derivation tree $\T_{\slp}$ from the root to the leaf corresponding to position $(N/2,N/2)$.
	Notice that position $(N/2,N/2)$ is contained in all $\lambda$ layers of $S$. 
    For every $i\in [\lambda/2]$, let $X_i$ be the last nonterminal in $P$ whose expansion contains $L_{2i}$.
    Since $V_1$ expands to $S$, which contains all layers of $S$, all $X_i$'s are well defined.
    We show that for every $i\in [\lambda/2 -1]$, $X_{i}$ strictly precedes $X_{i+1}$ in $P$, which implies that $|P| \ge \lambda/2 =  c \log N$, contradicting $d < c \log N$.
    \begin{claim}
    For every $i\in [\lambda/2 -1]$,  $X_i$ strictly precedes $X_{i+1}$ in $P$.
    \end{claim}
    \begin{proof}
    Since $L_{2i}$ contains $L_{2(i+1)}$ as a substring, it is clear that $X_i$ weakly precedes $X_{i+1}$ in $P$.
    Assume to the contrary that $X_i$ does not strictly precede $X_{i+1}$, which implies $X_i=X_{i+1}$.
    Let $X'=X_i=X_{i+1}$, and let the production rule of $X'$ be $\rho'(X')= B' \ovrt C'$. 
	Assume that the production rule of $X'$ is a vertical concatenation rule, as the other case (horizontal concatenation rule) is symmetric.
	Note that it is impossible for the production rule of $X'$ to be $\rho'(X')= \sigma$ for some terminal $\sigma$, as $\exp_{\G'}(X')$ contains $L_{2i}$.

    We will show that the bottom row of $B'$ contains a row of $L_{2i}$ that is not-too-close to the margins of $L_{2i}$ (See \cref{fig:lbproof}).
    Specifically, we will show that the bottom row of $B'$ contains the $\tilde{z}$-th row of $L_{2i}$ with $\tilde{z}$ being at distance at least $2 \Delta$ from the top and the bottom of $L_{2i}$.

    \begin{figure}[ht]
  \centering
  \tikzset{
pattern size/.store in=\mcSize, 
pattern size = 5pt,
pattern thickness/.store in=\mcThickness, 
pattern thickness = 0.3pt,
pattern radius/.store in=\mcRadius, 
pattern radius = 1pt}
\makeatletter
\pgfutil@ifundefined{pgf@pattern@name@_uvop7ga4w}{
\pgfdeclarepatternformonly[\mcThickness,\mcSize]{_uvop7ga4w}
{\pgfqpoint{0pt}{0pt}}
{\pgfpoint{\mcSize+\mcThickness}{\mcSize+\mcThickness}}
{\pgfpoint{\mcSize}{\mcSize}}
{
\pgfsetcolor{\tikz@pattern@color}
\pgfsetlinewidth{\mcThickness}
\pgfpathmoveto{\pgfqpoint{0pt}{0pt}}
\pgfpathlineto{\pgfpoint{\mcSize+\mcThickness}{\mcSize+\mcThickness}}
\pgfusepath{stroke}
}}

\tikzset{every picture/.style={line width=0.75pt}} %set default line width to 0.75pt        

\begin{tikzpicture}[x=0.75pt,y=0.75pt,yscale=-1,xscale=1,scale=0.8]
%uncomment if require: \path (0,453); %set diagram left start at 0, and has height of 453

%Shape: Rectangle [id:dp7118671920168457] 
\draw   (161,102) -- (461,102) -- (461,402) -- (161,402) -- cycle ;
%Shape: Rectangle [id:dp5047114149536884] 
\draw  [fill={rgb, 255:red, 155; green, 155; blue, 155 }  ,fill opacity=1 ] (411,102) -- (461,102) -- (461,402) -- (411,402) -- cycle ;
%Shape: Rectangle [id:dp45107642023058425] 
\draw  [fill={rgb, 255:red, 155; green, 155; blue, 155 }  ,fill opacity=1 ] (161,352) -- (411,352) -- (411,402) -- (161,402) -- cycle ;
%Shape: Rectangle [id:dp5995809338985583] 
\draw  [fill={rgb, 255:red, 155; green, 155; blue, 155 }  ,fill opacity=1 ] (161,102) -- (211,102) -- (211,352) -- (161,352) -- cycle ;
%Shape: Rectangle [id:dp9117869154674131] 
\draw  [fill={rgb, 255:red, 155; green, 155; blue, 155 }  ,fill opacity=1 ] (211,102) -- (411,102) -- (411,152) -- (211,152) -- cycle ;
%Shape: Rectangle [id:dp791556445284567] 
\draw   (361,152) -- (411,152) -- (411,352) -- (361,352) -- cycle ;
%Shape: Rectangle [id:dp6645619319333896] 
\draw   (211,152) -- (261,152) -- (261,302) -- (211,302) -- cycle ;
%Shape: Rectangle [id:dp49875439665200494] 
\draw   (211,302) -- (361,302) -- (361,352) -- (211,352) -- cycle ;
%Shape: Rectangle [id:dp6664672489625263] 
\draw   (261,152) -- (361,152) -- (361,202) -- (261,202) -- cycle ;
%Shape: Rectangle [id:dp6195211762885499] 
\draw (261,202) -- (361,202) -- (361,302) -- (261,302) -- cycle ;
%Shape: Largest Rectangle [id:dp6169641039773867] 
\draw   (128.67,89.5) -- (518.67,89.5) -- (518.67,415.5) -- (128.67,415.5) -- cycle ;
%Straight Lines [id:da3354333706202999] 
\draw [line width=1.5]  [dash pattern={on 5.63pt off 4.5pt}]  (128.67,221.17) -- (518.67,221.17) ;
% X' Lines [id:da1457844411887933] 
\draw    (128.67,78.5) -- (518.67,78.5) ;
\draw [shift={(518.67,78.5)}, rotate = 180] [color={rgb, 255:red, 0; green, 0; blue, 0 }  ][line width=0.75]    (10.93,-3.29) .. controls (6.95,-1.4) and (3.31,-0.3) .. (0,0) .. controls (3.31,0.3) and (6.95,1.4) .. (10.93,3.29)   ;
\draw [shift={(128.67,78.5)}, rotate = 0] [color={rgb, 255:red, 0; green, 0; blue, 0 }  ][line width=0.75]    (10.93,-3.29) .. controls (6.95,-1.4) and (3.31,-0.3) .. (0,0) .. controls (3.31,0.3) and (6.95,1.4) .. (10.93,3.29)   ;
% Arrow B' [id:da8164479243917545] 
\draw    (113.67,89.5) -- (113.67,220) ;
\draw [shift={(113.67,220)}, rotate = 270] [color={rgb, 255:red, 0; green, 0; blue, 0 }  ][line width=0.75]    (10.93,-3.29) .. controls (6.95,-1.4) and (3.31,-0.3) .. (0,0) .. controls (3.31,0.3) and (6.95,1.4) .. (10.93,3.29)   ;
\draw [shift={(113.67,89.5)}, rotate = 90] [color={rgb, 255:red, 0; green, 0; blue, 0 }  ][line width=0.75]    (10.93,-3.29) .. controls (6.95,-1.4) and (3.31,-0.3) .. (0,0) .. controls (3.31,0.3) and (6.95,1.4) .. (10.93,3.29)   ;
%Arrow C' [id:da21939588406123134] 
\draw    (113.67,223) -- (113.67,414) ;
\draw [shift={(113.67,415.5)}, rotate = 270] [color={rgb, 255:red, 0; green, 0; blue, 0 }  ][line width=0.75]    (10.93,-3.29) .. controls (6.95,-1.4) and (3.31,-0.3) .. (0,0) .. controls (3.31,0.3) and (6.95,1.4) .. (10.93,3.29)   ;
\draw [shift={(113.67,221.17)}, rotate = 90] [color={rgb, 255:red, 0; green, 0; blue, 0 }  ][line width=0.75]    (10.93,-3.29) .. controls (6.95,-1.4) and (3.31,-0.3) .. (0,0) .. controls (3.31,0.3) and (6.95,1.4) .. (10.93,3.29)   ;
%Straight Lines [id:da047359137657906225] 
\draw [line width=0.75]  [dash pattern={on 0.84pt off 2.51pt}]  (75.67,221.17) -- (128.67,221.17) ;
%Straight Lines [id:da08070472714165566] 
\draw [line width=0.75]  [dash pattern={on 0.84pt off 2.51pt}]  (75.67,102) -- (161,102) ;
%Arrow tilde z[id:da9704934116334487] 
\draw    (467.67,103.5) -- (467.67,220) ;
\draw [shift={(467.67,221.17)}, rotate = 270] [color={rgb, 255:red, 0; green, 0; blue, 0 }  ][line width=0.75]    (10.93,-3.29) .. controls (6.95,-1.4) and (3.31,-0.3) .. (0,0) .. controls (3.31,0.3) and (6.95,1.4) .. (10.93,3.29)   ;
\draw [shift={(467.67,101.5)}, rotate = 90] [color={rgb, 255:red, 0; green, 0; blue, 0 }  ][line width=0.75]    (10.93,-3.29) .. controls (6.95,-1.4) and (3.31,-0.3) .. (0,0) .. controls (3.31,0.3) and (6.95,1.4) .. (10.93,3.29)   ;
%Shape: Rectangle [id:dp05208988345527954] 
\draw  [fill={rgb, 255:red, 70; green, 70; blue, 70 }  ,fill opacity=1 ] (411,205.83) -- (461,205.83) -- (461,220.83) -- (411,220.83) -- cycle ;
%Straight Lines [id:da6870794613365898] 
\draw [line width=0.75]  [dash pattern={on 0.84pt off 2.51pt}]  (161,63.83) -- (161,102) ;
%Straight Lines [id:da09147490356300847] 
\draw [line width=0.75]  [dash pattern={on 0.84pt off 2.51pt}]  (211,40.83) -- (211,102) ;

% Text Node
\draw (308,252) node  [align=left] {\begin{minipage}[lt]{33pt}\setlength\topsep{0pt}
$\displaystyle L_{2( i+1)}$
\end{minipage}};
% Text Node
\draw (311,127) node   [align=left] {\begin{minipage}[lt]{19.95pt}\setlength\topsep{0pt}
$\displaystyle L_{2i}$
\end{minipage}};
% Text Node
\draw (316.17,60) node  [font=\large] [align=left] {\begin{minipage}[lt]{15pt}\setlength\topsep{0pt}
$\displaystyle X'$
\end{minipage}};
% Text Node
\draw (136,325) node   [align=left] {\begin{minipage}[lt]{68pt}\setlength\topsep{0pt}
$\displaystyle C'$
\end{minipage}};
% Text Node
\draw (141,158) node   [align=left] {\begin{minipage}[lt]{68pt}\setlength\topsep{0pt}
$\displaystyle B'$
\end{minipage}};
% Text Node
\draw (62,210) node [anchor=north west][inner sep=0.75pt]   [align=left] {$\displaystyle z$};
% Text Node
\draw (40,88) node [anchor=north west][inner sep=0.75pt]   [align=left] {$\displaystyle 2i\Delta $};
% Text Node
\draw (473,147) node [anchor=north west][inner sep=0.75pt]   [align=left] {$\displaystyle \tilde{z}$};
% Text Node
\draw (141,46) node [anchor=north west][inner sep=0.75pt]  [font=\small] [align=left] {$\displaystyle 2i\Delta $};
% Text Node
\draw (170,22.33) node [anchor=north west][inner sep=0.75pt]  [font=\small] [align=left] {$\displaystyle ( 2i+1) \Delta $};

\end{tikzpicture}
  \caption{A visualization of the proof of \cref{lem:lowerbound}.
        The outer rectangle represents the expansion of $X'$.
        The indices depicted to the left and above the outer rectangle represent absolute positions of the corresponding substrings within $S$.
        The dashed horizontal cut represents the derivation rule of $X'$. 
        Due to $X'$ being the last nonterminal containing $L_{2(i+1)}$ (represented as a light gray square), this cut must intersect $L_{2(i+1)}$, as depicted.
        The gray rectangles represent the $C_{N',M'}$ gadgets that form the frame of $L_{2i}$.
        Since the cut intersects $L_{2(i+1)}$, it intersects a vertical $C_{N',M'}$ gadget of $L_{2i}$ not-too-close to the boundaries ($L_{2i +1}$ functions as a 'buffer' layer between the $L_{2i}$ and $L_{2(i+1)}$, assuring that the cut is not within the first of the last $2\Delta$ rows of $L_{2i}$).
        The dark gray row is a row of some $C_{N',M'}$ gadget which is a substring of the bottom row of $B'$ (and the $\tilde{z}$-th row of the gadget).} \label{fig:lbproof}
\end{figure}
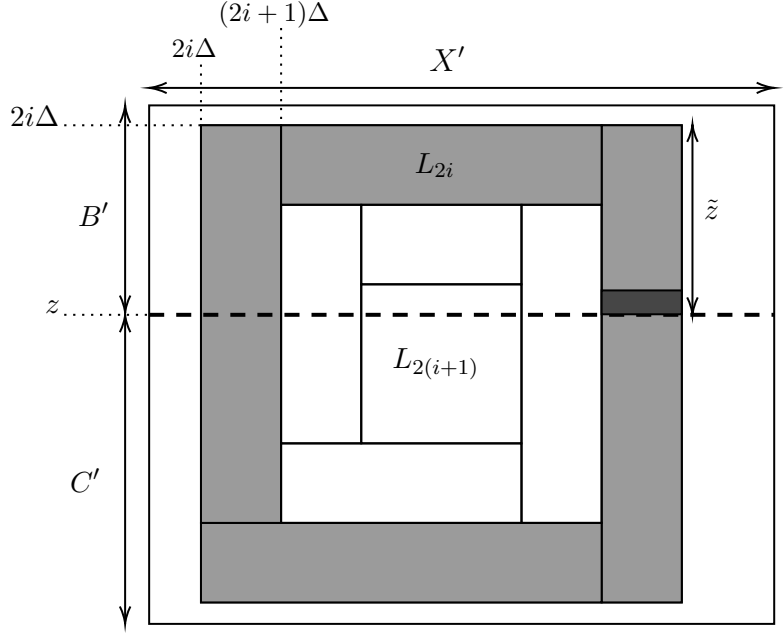

	Let the string $\exp_{\G'}(X')$ be $S_{X'} = S[x \dd x'][y \dd y']$.
    Since $S_{X'}$ contains $L_{2i}$, it holds that $(2 i\Delta \dd N-2i\Delta] \subseteq [x \dd x'] \cap [y \dd y']$.
    We claim that both $B'$ and $C'$ do not contain $L_{2i+2}$.
    If $X'=V_\ell$ and $V_{\ell+1} = B'$, we have that $B'$ does not contain $L_{2i+2}$ by the maximality of $V_{\ell}$ in the sequence $P$.
    Recall that the position $(N/2,N/2)$ is contained in $B'$ (and not in $C'$).
	This implies that $\exp_{\G'}(C')$ does not contain the unique square of $0$'s with dimensions $N-2\lambda\Delta \times N-2\lambda\Delta$ at the center of $S$, which is a necessary condition for containing any layer of $S$.
    Symmetric arguments can be made if $V_{\ell+1} = C'$.
    
    Intuitively, the fact that both $B'$ and $C'$ do not contain $L_{2i+2}$ means that the production rule $\rho'(X')= B' \ovrt C'$ 'cuts' the occurrence of $L_{2i+2}$ within $S_{X'}$.
	Formally, let $S_{B'}$ denote $\exp_{\G'}(B')$.
    It holds that $S_{B'}=S[x \dd z][y \dd y']$ for some $z\in ((2i+2)\Delta \dd N-(2i+2)\Delta]$.
    Consider the bottom row of $S_{B'}$ i.e. $S[z][y \dd y']$.
    Since the interval $[y \dd y']$ contains the horizontal span of $L_{2i}$, the bottom row $S[z][y \dd y']$ of $B'$ contains the $\tilde{z}$-th row of $L_{2i}$, with $\tilde{z} = z - 2i\Delta$.
    Indeed, we have $\tilde{z}\in (2\Delta \dd  N-(4i+2)\Delta]$, and in particular $\tilde{z}\in (2\Delta \dd  h-2\Delta]$ with $h=N-4i\Delta$ being the height of $L_{2i}$.
    It follows that $S[z][y \dd y']$ contains the $\tilde{z}$-th row of $C = C_{N- 4i\Delta,\Delta}$, which constitutes the rightmost $\Delta$ columns of $L_{2i}$.

    Let us focus on $R$, the $\tilde{z}$-th row of $C$ (See \cref{fig:lbproof2}).
    Recall that we took the largest $M'=2^{n}$ such that $M' (\log M' + 2) \le \Delta/2$.
    Due to the maximality of $M'$, we have $2M' (\log 2M' + 2) > \Delta/2$ which implies $M'\in \Omega(\frac{\Delta}{\log N}) = \Omega(\frac{N}{\log^2 N})$.
    Recall that $C$ is composed of two vertical blocks, each block consists of a vertical concatenation of $\ShiftBin_{M'}$ gadgets with dimensions $2M' \times M'(\log M' + 2)$.
    In the first block, every row except possibly the last $2M'-1$ rows is a row of some $\ShiftBin_{M'}$ gadget.
    In the second block, every row except the first $M'$ rows and the last $2M'-1$ rows is a row of some $\ShiftBin_{M'}$ gadget.
    We have $2M' \le M'\log M' \le \Delta/2 \le 2\Delta$ (The first inequality assumes $M'\ge 4$, which occurs for a sufficiently large $N$).
    It follows that $[2\Delta \dd  h- 2\Delta] \subseteq [2M' \dd h-2M']$ and therefore $\tilde{z} \in [2M' \dd h-2M']$.
    
    We have shown that $R$ contains a row of $\ShiftBin_{M'}$ in both blocks of $C$.
    Assume that $R$ contains the $r$-th row of a $\ShiftBin_{M'}$ in the left block of $C$.
    Due to the right block being vertically shifted by $M'$ compared to the left block, the row of $\ShiftBin_{M'}$ contained in $R$ from the second block of $C$ is $r' = r +M' \bmod 2M'$.
    One can verify that one of $r,r'$ must be in $[M'/2 \dd 3M'/2]$, assume without loss of generality that $r\in [M'/2 \dd 3M'/2]$.
    According to \cref{obs:2dframe1d} there is a 1D grammar deriving $R$ with $n'$ nonterminals such that  $n'$ is the number of nonterminals in $\G'$.
    Since $R$ is a superstring of the $r$-th row of $\ShiftBin_{M'}$, \cref{lem:shiftbincut} yields $n' \ge \min\{r,2M'-r\} \ge M'/2 \in \Omega(\frac{N}{\log^2 N})$, a contradiction to $n' \in o(\frac{N}{\log^2 N})$. 
    \end{proof}

    In conclusion, every 2D SLP deriving $S$ has either depth at least $c \log N$ or $\Omega(\frac{N}{\log^2 N}) = \Omega(|\G| \cdot \frac{N}{\log^3 N})$ nonterminals (recall that $|\G| \in \Oh(\log N)$), as required.
\end{proof}

\begin{figure}[ht] 
    \centering
     \tikzset{every picture/.style={line width=0.75pt}} %set default line width to 0.75pt        

\begin{tikzpicture}[x=0.75pt,y=0.75pt,yscale=-1,xscale=1]
%uncomment if require: \path (0,453); %set diagram left start at 0, and has height of 453

%Shape: Rectangle [id:dp2977262505148256] 
\draw  [fill={rgb, 255:red, 155; green, 155; blue, 155 }  ,fill opacity=1 ] (161,102) -- (261,102) -- (261,352) -- (161,352) -- cycle ;
%Straight Lines [id:da47366158334461883] 
\draw [line width=0.75]  [dash pattern={on 0.84pt off 2.51pt}]  (75.67,102) -- (161,102) ;
%Arrow tilde z [id:da13759757684192475] 
\draw    (153.67,103.5) -- (153.67,220.83) ;
\draw [shift={(153.67,220.83)}, rotate = 270] [color={rgb, 255:red, 0; green, 0; blue, 0 }  ][line width=0.75]    (10.93,-3.29) .. controls (6.95,-1.4) and (3.31,-0.3) .. (0,0) .. controls (3.31,0.3) and (6.95,1.4) .. (10.93,3.29)   ;
\draw [shift={(153.67,101.5)}, rotate = 90] [color={rgb, 255:red, 0; green, 0; blue, 0 }  ][line width=0.75]    (10.93,-3.29) .. controls (6.95,-1.4) and (3.31,-0.3) .. (0,0) .. controls (3.31,0.3) and (6.95,1.4) .. (10.93,3.29)   ;
%Straight Lines [id:da17970718141801834] 
\draw [line width=0.75]  [dash pattern={on 0.84pt off 2.51pt}]  (161,63.83) -- (161,102) ;
%Straight Lines [id:da22950243972589113] 
\draw [line width=0.75]  [dash pattern={on 0.84pt off 2.51pt}]  (261,40.83) -- (261,102) ;
%Straight Lines [id:da5134898885259935] 
\draw [line width=0.75]  [dash pattern={on 0.84pt off 2.51pt}]  (75.67,352) -- (161,352) ;
%Shape: Rectangle [id:dp07864969816081346] 
\draw  [fill={rgb, 255:red, 155; green, 155; blue, 155 }  ,fill opacity=1 ] (161,102) -- (211,102) -- (211,352) -- (161,352) -- cycle ;
%Shape: Dark line [id:dp0810119280456928] 
\draw  [fill={rgb, 255:red, 70; green, 70; blue, 70 }  ,fill opacity=1 ] (161,205.83) -- (261,205.83) -- (261,220.83) -- (161,220.83) -- cycle ;
%Shape: Rectangle [id:dp6400567678561597] 
\draw  [fill={rgb, 255:red, 255; green, 255; blue, 255 }  ,fill opacity=1 ] (161,322) -- (211,322) -- (211,352) -- (161,352) -- cycle ;
%Shape: Rectangle [id:dp4746307173337726] 
\draw  [fill={rgb, 255:red, 255; green, 255; blue, 255 }  ,fill opacity=1 ] (211,102) -- (261,102) -- (261,122) -- (211,122) -- cycle ;
%Shape: Rectangle [id:dp8561788758442769] 
\draw  [fill={rgb, 255:red, 255; green, 255; blue, 255 }  ,fill opacity=1 ] (211,342) -- (261,342) -- (261,352) -- (211,352) -- cycle ;
%Arrow M' Lines [id:da4928977468537552] 
\draw    (266.67,105) -- (266.67,121) ;
\draw [shift={(266.67,122)}, rotate = 270] [fill={rgb, 255:red, 0; green, 0; blue, 0 }  ][line width=0.08]  [draw opacity=0] (8.93,-4.29) -- (0,0) -- (8.93,4.29) -- cycle    ;
\draw [shift={(266.67,102)}, rotate = 90] [fill={rgb, 255:red, 0; green, 0; blue, 0 }  ][line width=0.08]  [draw opacity=0] (8.93,-4.29) -- (0,0) -- (8.93,4.29) -- cycle    ;
%Straight Lines [id:da07073796626669537] 
\draw  [dash pattern={on 0.84pt off 2.51pt}]  (161,322) -- (286.67,322) ;
%Straight Lines [id:da39867077626943115] 
\draw  [dash pattern={on 0.84pt off 2.51pt}]  (161,352) -- (286.67,352) ;
%Straight Lines [id:da44344299399609366] 
\draw    (286.67,325) -- (286.67,349) ;
\draw [shift={(286.67,352)}, rotate = 270] [fill={rgb, 255:red, 0; green, 0; blue, 0 }  ][line width=0.08]  [draw opacity=0] (8.93,-4.29) -- (0,0) -- (8.93,4.29) -- cycle    ;
\draw [shift={(286.67,322)}, rotate = 90] [fill={rgb, 255:red, 0; green, 0; blue, 0 }  ][line width=0.08]  [draw opacity=0] (8.93,-4.29) -- (0,0) -- (8.93,4.29) -- cycle    ;

% Text Node
\draw (310,246.25) node  [color={rgb, 255:red, 255; green, 255; blue, 255 }  ,opacity=1 ] [align=left] {\begin{minipage}[lt]{68pt}\setlength\topsep{0pt}
Layer $\displaystyle 2( i+1$)
\end{minipage}};
% Text Node
\draw (43,93) node [anchor=north west][inner sep=0.75pt]   [align=left] {$\displaystyle 2i\Delta $};
% Text Node
\draw (142,147) node [anchor=north west][inner sep=0.75pt]   [align=left] {$\displaystyle \tilde{z}$};
% Text Node
\draw (4,343) node [anchor=north west][inner sep=0.75pt]   [align=left] {$\displaystyle N-( 2i) \Delta $};
% Text Node
\draw (283.33,111.75) node   [align=left] {\begin{minipage}[lt]{15.41pt}\setlength\topsep{0pt}
$\displaystyle M'$
\end{minipage}};
% Text Node
\draw (320.5,334.75) node   [align=left] {\begin{minipage}[lt]{39.69pt}\setlength\topsep{0pt}
$\displaystyle <\!2M'$
\end{minipage}};
% Text Node
\draw (149,39.67) node [anchor=north west][inner sep=0.75pt]  [font=\small] [align=left] {$\displaystyle 2i\Delta $};
% Text Node
\draw (228,24) node [anchor=north west][inner sep=0.75pt]  [font=\small] [align=left] {$\displaystyle ( 2i+1) \cdot \Delta $};

\end{tikzpicture}
        \caption{
            A demonstration of the $\tilde{z}$ row of $C=C_{\Delta,N-4\Delta}$ (The third padding block is ignored here).
            The white rectangle represents the 'padding' zeros added to each block of $C$, the grey area in each row is tiled with occurrences of $\ShiftBin_b$.
            Since $\tilde{z}$ is at least $2b$ indices away from the top/bottom of the gadget, the $\tilde{z}$-th row is contained in the gray areas of both rows.
            }
        \label{fig:lbproof2}
\end{figure}
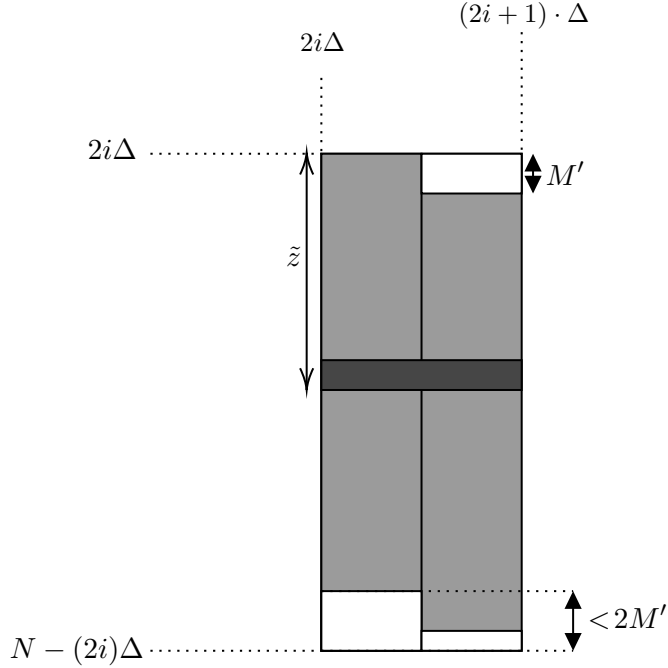

\section{Balancing Two-Dimensional SLPs by Allowing Holes}

As shown in \cref{lem:lowerbound}, it is not always possible to balance 2D SLPs while keeping their size small. To overcome this limitation, we next introduce \emph{two-dimensional tree straight-line programs} (2D TSLPs), or \emph{2D SLPs with holes}, which are a generalization of 2D SLPs that allow \emph{contexts} in the derived objects (For clarity, standard 2D strings are called \emph{ground} 2D strings). Intuitively, a context is a 2D string $S$ defined over a \emph{frame} (a rectangle from which a smaller $p \times q$ rectangle has been removed). It corresponds to a function associating a 2D string $T$ of size $p \times q$ to the string $S(T)$, obtained by filling the hole of $S$ with $T$. A 2D TSLP also allows extra productions, corresponding to context evaluation, composition, and concatenation of a context with a ground 2D string (see \cref{fig:operations}). We then apply a theorem from~\cite{DBLP:journals/jacm/GanardiJL21}, stating that from any SLP, one can construct an equivalent balanced TSLP having the same size. Using this balanced representation, we then describe a data structure providing efficient random access on a 2D string.

The result of~\cite{DBLP:journals/jacm/GanardiJL21} is based on the more general setting of \emph{(many-sorted) algebras}. Informally, a many-sorted algebra is a set of typed objects and functions that respect a particular \emph{syntax}, namely the rules under which objects and functions can be combined into valid expressions, depending on object types and function arities, but not on the values of such expressions. Because a SLP is defined as a set of nonterminals and productions formed by valid expressions, one relies only on the syntax to construct the derivation tree; consequently, properties such as size and depth are defined independently of the SLP's actual evaluation in a particular algebra. The authors define SLPs and TSLPs within this framework (in particular, contexts are defined as syntactically valid expressions where a term has been replaced by a ``hole'') and demonstrate that any SLP can be balanced into a TSLP that derives the same expression and having the same asymptotic size as the original SLP. The result is extremely general, in the sense that it does not depend on the actual objects considered: when evaluating the expression in an algebra, the TSLP is guaranteed to evaluate to the same value as the original SLP, yet the derivation tree is guaranteed to have logarithmic depth. The formal distinction between syntax and algebra is particularly relevant for the second main result of~\cite{DBLP:journals/jacm/GanardiJL21}, stating conditions under which, from the balanced TSLP, one can get back to an equivalent SLP with the same size and depth (hence, also balanced). Informally, this is possible when there exists a finite set of formulas such that every context is equivalent, in the algebra, to one of these formulas. 

In this work, we ignore this distinction and simply define many-sorted algebras as collections of sets and functions, the underlying syntax being implicit. This is enough to reformulate the first results from~\cite{DBLP:journals/jacm/GanardiJL21} about constructing balanced TSLPs from SLPs. For more details on this, we also refer the reader to~\cite{ganardiUniversal2019}, which contains the main ideas of the latter article in a simpler setting, and for a more general introduction on many-sorted algebra to~\cite{adamekAlgebraicTheoriesCategorical2010,wechlerUniversal1992}.

\para{Definitions.}

Let $\Scal$ be a finite set, $A=\bigsqcup_{s\in \Scal} A_s$ be a disjoint union of sets indexed by $\Scal$ and $\Gamma$ be a set of functions such that for every $f\in \Gamma$, there exist $p_1,\dots,p_{\rk},q\in \Scal$ such that $f:A_{p_1}\times \dots \times A_{p_{\rk}}\rightarrow A_q$. In this case, we write $\rank(f)=\rk$, $\type(f)=(p_1\dots p_{\rk},q)$ and $\sort(f)=q$. Motivated by~\cite{DBLP:journals/jacm/GanardiJL21}, we call the pair $(A,\Gamma)$ a \emph{many-sorted algebra}.

A \emph{straight-line program} over a many-sorted algebra $(A,\Gamma)$ is a tuple $\slp=(\nonterm,\start,\pro)$, such that:
\begin{itemize}
    \item $\nonterm=\bigsqcup_{s\in \Scal} \nonterm_s$ 
    is a finite set of \emph{nonterminals}, with $\nonterm\cap \Gamma=\emptyset$.
    \item $\start\in \nonterm$ is called the \emph{starting nonterminal}.
    \item $\pro$ is a mapping on nonterminals, where for each $X\in \nonterm$, one of the following holds:
    \begin{itemize}
        \item $\pro(X) =  a$ where $a\in A_q$ for some $q\in \Scal$,
        \item $\pro(X) = f(T_1,\dots,T_{\rk})$ where $X\in \nonterm_q$ for some $q\in \Scal$, $f\in \Gamma$ with $\type(f)=(p_1\dots p_{\rk},q)$, $T_i\in A_{p_i}\sqcup \nonterm_{p_i}$ for every $i\in [\rk]$.
   \end{itemize}
    Each pair $(X,\pro(X))$ for $X\in \nonterm$ is called a \emph{production}, with $X$ being its \emph{left-hand side} and $\pro(X)$ its \emph{right-hand side}. Additionally, the relation on $\nonterm$ defined by $X\le Y$ if $Y$ appears in the right-hand side of $X$ must be acyclic.
\end{itemize}

From the acyclicity condition, it follows that for each nonterminal $X \in \nonterm$, there is a unique element $\exp_{\slp}(X)$ derived from $X$ by recursively replacing each nonterminal in the right-hand side of its production by the corresponding derived element.
The element derived by the SLP is denoted by $\exp(\slp) = \exp_{\slp}(\start)$.
We define the \emph{size} $|\slp|$ as the total number of symbols on the right-hand sides of all productions, and the \emph{depth} of $\slp$ as the depth of its derivation tree.

Given a function $f\in \Gamma$ with $\type(f)=(p_1\dots p_{\rk},q)$, $i\in [\rk]$, and $a_1\in A_{p_1}, \dots, a_{i-1}\in A_{p_{i-1}}$, $a_{i+1}\in A_{p_{i+1}},\dots, a_{\rk}\in A_{p_{\rk}}$, we define the \emph{context} $f(a_1,\dots,a_{i-1},*,a_{i+1},\dots,a_{\rk}):A_{p_i}\rightarrow A_q$ as the mapping taking $b\in A_{p_i}$ to $f(a_1,\dots,a_{i-1},b,a_{i+1},\dots,a_{\rk})$.

A TSLP over a many-sorted algebra $(A,\Gamma)$ is a tuple $\slp=(\nonterm,\start,\pro)$ where:
\begin{itemize}
    \item $\nonterm=\nonterm_G\sqcup \nonterm_C$ is a finite set of \emph{nonterminals}, with $\nonterm\cap \Gamma=\emptyset$, partitioned into two disjoint sets:
    \begin{itemize}
        \item $\nonterm_G=\bigsqcup_{s\in \Scal} \nonterm_{G,s}$ is a set of \emph{ground nonterminals}.
        \item $\nonterm_C=\bigsqcup_{(p,q)\in \Scal^2} \nonterm_{C(p,q)}$ is a set of \emph{context nonterminals}, where each $X\in \nonterm_{C(p,q)}$ represents a context mapping an element of $A_p$ to an element of $A_q$. In that case, we can also write $X(*)$ instead of $X$.
    \end{itemize}
    \item $\start\in \nonterm_G$ is called the \emph{starting nonterminal}.
    \item $\pro$ is a mapping on nonterminals, where for each $X\in \nonterm$: \begin{itemize}
        \item If $X\in \nonterm_{G,q}$ for some $q\in \Scal$, one of the following holds:
        \begin{itemize}
            \item $\pro(X)=a$ where $a\in A_q$,
            \item $\pro(X)=f(T_1,\dots,T_{\rk})$ where $f\in \Gamma$ with $\type(f)=(p_1\dots p_{\rk},q)$, $T_i\in A_{p_i}\sqcup \nonterm_{G,p_i}$ for every $i\in [\rk]$.
        \item $\pro(X)=Y(T)$ where $Y\in \nonterm_{C(p,q)}$ for some $p\in \Scal$, and $T\in A_p\sqcup \nonterm_{G,p}$.
        \end{itemize}
        \item If $X\in \nonterm_{C(p,q)}$ for some $p,q\in \Scal$, one of the following holds:
        \begin{itemize}\raggedright
            \item $\pro(X(*))=f(a_1,\dots,a_{i-1},*,a_{i+1},\dots,a_{\rk})$ where $f\in \Gamma$ with $\type(f)=(p_1\dots p_{\rk},q)$, $p=p_i$ and $a_j\in A_{p_j}$ for every $j\in [\rk]\setminus\{i\}$. 
            \item $\pro(X(*))=f(T_1,\dots,T_{i-1},Y(*),T_{i+1},\dots,T_{\rk})$ where $f\in \Gamma$ with $\type(f)=(p_1\dots p_{\rk},q)$, $Y\in \nonterm_{C(p,p_i)}$, and $T_j\in A_{p_j}\sqcup \nonterm_{G,p_j}$ for every $j\in [\rk]\setminus\{i\}$.
            \item $\pro(X(*))=Y(Z(*))$ where $Y\in \nonterm_{C(s,q)}$ for some $s\in \Scal$, and $Z\in \nonterm_{C(p,s)}$.
    \end{itemize}
    \end{itemize}
		As for SLPs, the relation on $\nonterm$ defined by $X\le Y$ if $Y$ appears in the right-hand side of $X$ must be acyclic.
    \end{itemize}

We define $\exp(\slp)$, $|\slp|$ and $\depth(\slp)$ as for SLPs. Note that, since $\Scal\in\nonterm_G$, $\exp(\slp)$ is an element of $A$.

Our definition of SLPs and TSLPs requires the productions to be in a Chomsky normal form. In that sense, it differs from the one given in~\cite{DBLP:journals/jacm/GanardiJL21}, where productions can have a right-hand side of depth more than $1$, and where the depth is defined accordingly. However, it is easy to see that any general SLP or TSLP can be transformed into an equivalent one in a normal form, with equal depth and size increased only by a constant factor.

\begin{theorem}[{\cite[Theorem 3.19]{DBLP:journals/jacm/GanardiJL21}}]\label{thm:balance slp}
Let $(A,\Gamma)$ be a many-sorted algebra and let $\slp$ be a SLP over $(A,\Gamma)$ deriving $t\in A$. One can in time $\Oh(|\slp|)$ compute an equivalent TSLP over $(A,\Gamma)$ of size $\Oh(|\slp|)$ and depth $\Oh(\log |t|)$.
\end{theorem}

In fact, under the condition that functions from $\Gamma$ have bounded rank, one can generalize Theorem~\ref{thm:balance slp} in order to balance TSLPs. Formally, one can state the following:

\begin{theorem}[{\cite[Remark 3.27]{DBLP:journals/jacm/GanardiJL21}}]\label{thm:balance tslp}
Let $(A,\Gamma)$ be a many-sorted algebra and let $\slp$ be a TSLP over $(A,\Gamma)$ deriving $t\in A$. One can in time $\Oh(|\slp|)$ compute an equivalent TSLP over $(A,\Gamma)$ of size $\Oh(\rk\cdot|\slp|)$, and depth $\Oh(\log |t|)$, where $\rk$ is the maximal rank of a function in $\Gamma$.
\end{theorem}

\para{Balancing 2D SLPs with holes.}

We define 2D SLPs with holes as TSLPs over the many-sorted algebra $(A,\Gamma)$ where $A=\bigsqcup_{(n,m)\in \N^2} \Sigma^{n\times m}$ and $\Gamma$ contains the horizontal and vertical concatenation functions (formally, $\Gamma$ contains one horizontal concatenation for each height, and one vertical concatenation for each width, but we omit this distinction). In that setting, context nonterminals can be seen as 2D strings with a hole, that can be filled with another 2D string of appropriate dimensions.

\begin{example}
        Consider the following small 2D SLP with holes over the alphabet $\Sigma\coloneqq \{0,1\}$. Ground nonterminals are $A,B,C,T$ and there is one context nonterminal $H(*)$.
        The productions are:
        \begin{itemize}
        \item $\pro(A)\coloneqq 0$
        \item $\pro(B)\coloneqq 1$
        \item $\pro(C)\coloneqq A \ohrz B $
        \item $\pro(H(*))\coloneqq * \ovrt C $
        \item $\pro(T)\coloneqq H(C)$
        \end{itemize}

		We have $\exp(C)= 01$, $\exp(H(*))$ is a 2D string with hole on top of $01$, and finally \[\exp(T)= \begin{bmatrix} 0 & 1 \\ 0 & 1 \end{bmatrix}.\]
    \end{example}

From \cref{thm:balance tslp}, and noticing that $\Gamma$ only contains elements of rank $2$, we obtain the following (notice that standard 2D SLPs constitute a particular case of 2D SLPs with holes):

\begin{theorem}\label{thm:2d tslp}
	Given a 2D SLP $\G$ with holes deriving a string $T \in \Sigma^{N \times M}$, one can construct  in time $\Oh(|\slp|)$ an equivalent 2D SLP with holes of size $\Oh(|\G|)$ and depth $\Oh(\log (NM))$.
\end{theorem}

With \Cref{thm:2d tslp}, we can show how to obtain a random access structure that improves the bounds of De and Kempa~\cite{de2025optimalrandomaccessconditional} by two logarithmic factors. 

\begin{theorem}
\label{thm:random}
Fix an arbitrary constant $\epsilon>0$.
Given a 2D SLP $\G$ deriving a 2D string $T \in \Sigma^{N \times M}$, one can construct
in time $\Oh(|\slp|\cdot \log^{\epsilon}(NM))$ a data structure of size $\Oh(|\slp|\cdot \log^{\epsilon}(NM))$
that allows accessing any $T[x][y]$ in $\Oh(\log(NM)/\log\log(NM))$ time.
\end{theorem}

\begin{proof}
We start by applying \Cref{thm:2d tslp} to obtain an equivalent 2D SLP with holes of size
$\Oh(|\G|)$ and depth $\Oh(\log (NM))$, denoted $\G'$. For every nonterminal of $\G'$
we store the size of its derived string. Additionally, for every context nonterminal of $\G'$
we store the position of its hole. Then, it is straightforward to implement random access
by navigating down in the grammar from the starting nonterminal, in every step using
the stored information about the size of the derived strings (and possibly the position of the hole)
for every nonterminal participating in the currently considered production to decide
in constant time where to continue. 

Let us write $\G'=(\nonterm',\start',\pro')$. To improve the random access time, we proceed
similarly as in \cite[Theorem 2]{DBLP:conf/esa/BelazzouguiCPT15}, namely we define new productions $\tilde{\pro}$ as $\tilde{\pro}(X)=\pro'^K(X)$, expanding the grammar by $K=\lfloor\log\log^{\epsilon/3}(NM)\rfloor$ levels. Each right-hand side now has size at most $\log^{\epsilon/3}(NM)$, and the depth of the grammar has decreased to $\Oh(\log(NM)/\log\log(NM))$. 

In the production rule for a nonterminal $V$ (whether ground or context), each symbol corresponds to an axis-aligned domain, which is either a rectangle or a frame. Specifically, rectangular domains correspond to ground nonterminals (or to the hole of $V$, if $V$ is a context nonterminal), while frame domains correspond to context nonterminals. Together, these domains partition the $\h(V)\times \w(V)$ area into $b=\Oh(\log^{\epsilon/3}(NM))$ smaller regions. 
We build a $b \times b$ grid by creating vertical and
horizontal lines containing the sides of every rectangle (including the removed rectangle
for a frame). We store the $x$ or $y$ coordinates for each line in a fusion tree (Lemma~\ref{lem:fusion}).
This is illustrated in \cref{fig:random-access}.
For each of the $b^{2}$ cells, we store a pointer to the rectangle or the frame that fully
contains it, and the relative position of the cell inside that rectangle or frame. Then,
given a position $(x,y)$, we first use the fusion trees to locate in constant time the cell
that contains $(x,y)$. Then, we retrieve the rectangle or the frame that contains $(x,y)$
using the stored pointer. This gives us the ground or context nonterminal in which
we should continue the descent, together with the new position $(x',y')$.
The overall size of the structure is $\Oh(|\G'|\cdot b^{2})=\Oh(|\G|\cdot \log^{\epsilon}(NM))$.
\end{proof}

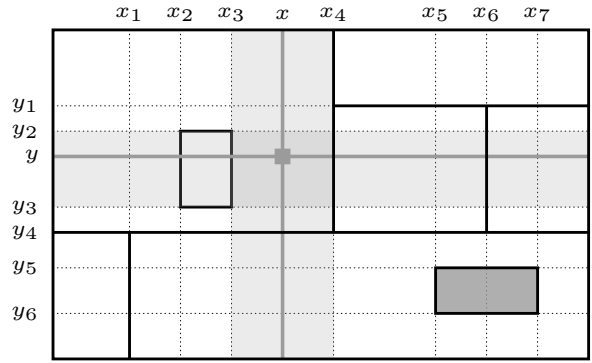
\begin{figure}[ht]
\begin{center}
\definecolor{mygrey}{RGB}{155,155,155}

\resizebox{0.5\textwidth}{!}{%
\begin{circuitikz}

	\tikzstyle{main lines} = [line width=0.8pt]
    \tikzstyle{grid lines} = [line width=0.2pt, densely dotted, line cap=round]    
	\tikzstyle{query lines} = [line width=1pt, color=mygrey]
    \tikzstyle{lbl}        = [font=\tiny, anchor=south]
    \tikzstyle{lbl-side}   = [font=\tiny, anchor=east]

	\foreach \pos/\i in {4.5/1, 5/2,5.5/3, 6.5/4, 7.5/5, 8.0/6, 8.5/7} {
        \coordinate (x\i) at (\pos, 0);  
        \draw[grid lines] (\pos, 11.5) -- (\pos, 8.25); 
        \node[font=\tiny] at (\pos, 11.65) {$x_\i$};
    }

	\foreach \pos/\i in {10.75/1, 10.5/2, 9.75/3, 9.5/4, 9.15/5, 8.7/6} {
        \coordinate (y\i) at (0, \pos);
        \draw[grid lines] (3.75, \pos) -- (9, \pos);
        \node[font=\tiny, anchor=east] at (3.75, \pos) {$y_\i$};
	}

    \coordinate (left) at (3.75, 0);
    \coordinate (right) at (9.0,  0);
    \coordinate (top)  at (0, 11.5);
    \coordinate (bot)  at (0, 8.25);

	\coordinate (p)  at (6, 0);
    \coordinate (q)  at (0, 10.25);
	\node[font=\tiny] at (p |- 0, 11.65) {$x$};
	\node[font=\tiny, anchor=east] at (3.75, 0 |- q) {$y$};    
	\node at (p |- q) [
    rectangle,      
    draw=mygrey, 
    fill=mygrey,
    minimum size=4pt,
    inner sep=0pt 
] {};

    % Holes
    \draw [main lines] 
        (x2 |- y2) rectangle (x3 |- y3);

	\draw [fill=mygrey, fill opacity=0.8, main lines] 
        (x5 |- y5) rectangle (x7 |- y6);

	\fill [color=mygrey, fill opacity=0.2] (x3 |- top) rectangle (x4 |- bot);
	\fill [color=mygrey, fill opacity=0.2] (right |- y2) rectangle (left |- y3);

\draw[query lines] (p |- top) -- (p |- bot); 
\draw [query lines] (left |- q) -- (right |- q);

\draw [main lines] (left |- top) rectangle (right |- bot);
\draw [main lines] (x4 |- y1) -- (right |- y1); 
\draw [main lines] (left |- y4) -- (right |- y4); 
\draw [main lines] (x1 |- y4)  -- (x1 |- bot);   
\draw [main lines] (x4 |- top)  -- (x4 |- y4);  
\draw [main lines] (x6 |- y1)  -- (x6 |- y4);

\end{circuitikz}
}
\caption{The $\h(V)\times \w(V)$ rectangle is partitioned into axis-aligned rectangles and frames (solid boundaries). The dark gray region indicates a hole.  These boundaries induce a non-uniform grid (dotted lines) with coordinates $x_i, y_i$ stored in a fusion tree. A query point $(x,y)$ (black square) falls into the grid cell defined by the intersection of the vertical and horizontal intervals $[x_3, x_4]$ and $[y_2, y_3]$ (light gray). These intervals are identified via predecessor queries on the fusion tree. This cell maps to a unique child nonterminal (here, the top-left frame) used for the next step of the descent. }
\end{center}    
\label{fig:random-access}
\end{figure}

\bibliographystyle{plain}
\bibliography{biblio}

\end{document}